\documentclass{amsart}

\usepackage{amsmath,amsthm,amsfonts,eucal,eufrak}

\usepackage{latexsym}
\usepackage{amssymb}

\usepackage[active]{srcltx}
\usepackage[dvips]{epsfig}

\newcommand{\cA}{{\mathcal{A}}}

\def\bm{\big |}

\def\bn{\big \|}

\newcommand{\nn}{\nonumber}

\newcommand{\cB}{{\mathcal{B}}}

\newcommand{\cG}{{\mathcal{G}}}

\newcommand{\cM}{{\mathcal{M}}}

\newcommand{\cF}{{\mathcal{F}}}

\newcommand{\cX}{{\mathcal{X}}}

\newcommand{\IR}{{\mathbb{R}}}

\def\beeq{\begin{equation}}
\def\eneq{\end{equation}}

\newcommand{\be}{\begin{eqnarray}}
\newcommand{\ee}{\end{eqnarray}}

\newcommand{\mes}{{\rm mes}}

\newcommand{\Var}{{\mathrm{Var}}}

\newcommand{\pr}{\mathbb{P}}
\newcommand{\expc}{\mathbb{E}}

\newtheorem{theorem}{Theorem}[section]
\newtheorem{lemma}[theorem]{Lemma}
\newtheorem{cor}[theorem]{Corollary}
\newtheorem{prop}[theorem]{Proposition}

\theoremstyle{definition}
\newtheorem{defi}[theorem]{Definition}

\theoremstyle{remark}
\newtheorem{remark}[theorem]{Remark}

\newtheorem*{thma}{Theorem A}
\newtheorem*{thmb}{Theorem B}
\newtheorem*{thmc}{Theorem C}

\numberwithin{equation}{section}

\begin{document}
\title[Fluctuations and Localization Length for Random Band GOE Matrix.]{Fluctuations and Localization Length for Random Band GOE Matrix.}
\author{Michael Goldstein }

\address{Dept.\ of Mathematics, University of Toronto, Toronto, Ontario, Canada M5S 1A1}

\email{gold@math.toronto.edu}
\thanks{Supported in part by  NSERC}

\date{}

\begin{abstract} We prove that GOE random band matrix localization length is $\le C\left(\log W\right)^3 W^2$, where $W$ is the width of the band and $C$ is an absolute constant.
Our method consists of Green function edge-to-edge vector action approach to the Schenker method. That allows to split and decouple the action, so that
it becomes transparent that \emph{the magnitudes of two consecutive Schur complements vector actions can not be both larger than an absolute constant}. That is the central technological ingedient of the method. It comes from rather involved estimates $($ the main estimates of the metod $)$, in combination with an equation relating two magnitudes in question. We call the latter \emph{recurrence equation}. The method results in the \emph{lower bound of the variance of the $\log$--norm of the vector action at $\gtrsim NW^{-1}$}, where $N$ is the total number of GOE blocks, condition $N\lesssim W^D$ with an absolute constant $D\gg 1$ applies.
\end{abstract}

\maketitle
\tableofcontents

\section{Introduction}
We consider the following random band matrix
\begin{equation}
\label{eq:gdef1}
H_{[1,N]}=\begin{bmatrix}
V_1 & T_1 & 0 & \cdots & \cdots & 0\\
T_1^t & V_2 & T_2 & 0&\cdots & 0\\
0 & T^t_2 & V_3 & T_3&0\cdots & 0\\
\cdots&\cdots&\cdots&\cdots&\cdots&\cdots\\
\cdots&\cdots&\cdots&\cdots&\cdots&\cdots\\
0 & \cdots & \cdots &0 & T^t_{N-1} & V_N\end{bmatrix},
\end{equation}
Here $V_k,T_\ell$ are independent $W\times W$ random matrices. The matrices $V_k = (v_{k, p, q})_{1 \leq p, q \leq W}$ are real symmetric from the GOE, i.e.
\be
\nn d\pr\big(V_k\big) = \prod_{1 \leq p\leq W}\sqrt{\frac{W}{2\pi}}e^{-\frac{Wv_{k,p,p}^2}{4}}dv_{k,p, p} \prod_{1 \leq p< q \leq W}\sqrt{\frac{W}{2\pi}}e^{-\frac{Wv_{k,p, q}^2}{2}}dv_{k,p, q}:=\\
\label{eq:GOEmatrix}\\
\nn \phi(V_k)dV_k,\qquad\qquad\qquad\qquad\qquad\qquad\qquad
\ee
where
\[
dV=\prod_{1 \leq p\le q \leq W}dv_{p, q}
\]
stands for the Lebesgue measure in $\IR^{{W(W+1)\over 2}}$, as we identify every real symmetric matrix $V=(v_{p, q})_{1 \leq p,q \leq W}$ with the vector $(v_{p, q})_{1 \leq p\le q \leq W}\in\IR^{{W(W+1)\over 2}}$.
The matrices $T_k = (\xi_{k, p, q})_{1 \leq p, q \leq W}$ are $W\times W$ real random Gaussian non-symmetric, i.e.
\be
d\pr\big(T_k\big) =\prod_{1 \leq p,q \leq W}\sqrt{\frac{W}{2\pi}}e^{-\frac{W\xi_{k,p, q}^2}{2}}d\xi_{k,p, q}:=\rho(T_k)dT_k
\label{eq:Gaussianonsym}
\ee
where
\[
dT=\prod_{1 \leq p,q \leq W}d\xi_{p, q}
\]
stands for the Lebesgue measure in $\IR^{W^2}$, as we identify every real matrix $T=(\xi_{p, q})_{1 \leq p,q \leq W}$ with the vector $(\xi_{p, q})_{1 \leq p, q \leq W}\in\IR^{W^2}$,
\begin{equation}
\label{eq:f_indep}
d\mathbb{P}_{[1, N]}{\big(V_1,V_2, \ldots,T_1,T_2,\dots\big)} = \prod_{k=1}^{N}d\pr\big(V_k\big) \prod_{k=1}^{N-1}d\pr\big(T_k\big)
\end{equation}
We use also the equations
\[
d\pr\big(V\big)=\mu_W e^{-{W\over 4}Tr V^2}dV,
\]
\[
d\pr\big(T\big)=\nu_W e^{-{W\over 2}Tr T^tT}dT,
\]
where $\mu_W,\nu_W$ are the normalizing factors.

Fix $E\in\IR$. Denote $G_{[1,N]} = \big[H_{[1, N]}-E\big]^{-1}$ the Green function. Denote $\cM_W$ the space of all $W\times W$ matrices. For $1\le m,n\le N$ introduce also the following notations
\[G_{[1, N]}(m; n):= \left(G_{[1, N]}(p, q)\right)_{(m-1)W+1\leq p\leq mW; (n-1)W+1\leq q \leq nW}\in \cM_W
\]
We call $G_{[1, N]}(1; N)$ \emph{Green function edge--to--edge matrix}.
\begin{thma}\label{thm:A} There exists a large absolute constant $D_0\gg 1$ such that for $N\ge W^{D_0}$ the following statement holds. There exists a set $\cB_N$ such that
\[
\pr\big[\cB_N\big]\le \exp\left[-c_0{N\over W^{D_0}}\right]
\]
where $c_0>0$ is an absolute constant, and off this set set 
\be\label{eq:LexponentTHMA}
\bn G_{[1,N]}(1;N)\bn<e^{-{cN\over W\left(\log W\right)^3}},
\ee
where $c>0$ is an absolute constant.
\end{thma}

\begin{thmb}\label{thm:B}
There exists a large absolute constant $D_0\gg 1$ such that for $W^2\le N\le W^{D_0}$ the following statement holds. Take arbitrary unit vector $f\in\IR^W$. Then
\be\label{eq:vartheorem}
\Var\left[\log\bn G_{[1,N]}(1;N)f\bn\right]\ge {cN\over W}
\ee
where $c>0$ is an absolute constant.
\end{thmb}

\begin{remark}\label{rem:initialscale}
$(1)$ We would like to remark here on why we need more estimates on top of the variance in Theorem B. The reason for that is a trivial objection for such general conclusion. This objection is possible to eliminate for the setting in question. The objection is as follows. Take \emph{a non-positive random variable $\xi$}. Assume that $\Var \xi=\mathfrak S^2\gg 1$. It is \emph{not true, in general}, that
\begin{gather*}
\expc\left[e^\xi\right] \le 
e^{-c\mathfrak S^2},
\end{gather*}
with an absolute constant $c>0$. A trivial example of a random variable taking two values only
\begin{gather*}
\begin{cases}
\pr\left\{\xi=-\mathfrak S^{{3\over 2}}\right\}=p:=\mathfrak S^{-1},\\
\pr\left\{\xi=0\right\}=q:=1-p
\end{cases}
\end{gather*}
results
\begin{gather*}
\Var \xi=\left(\mathfrak S^2-\mathfrak S\right)\sim \mathfrak S^2,\\
\expc\left[e^\xi\right]=p e^{-\mathfrak S^{{3\over 2}}}
\end{gather*} 
The reason for this kind of outcome is \emph{the lack of a "classical" large deviation estimate} for $\xi$. The latter means that the estimate in question comes from the projection of one with $\xi=\sum_{1\le j\le N}\xi_k$ and with $\xi_k$ being independent, not necessarily equaly distributed, with "classical" Bernstein condition for the moments, see for instance ~\cite{Pe-75}:
\begin{gather*}
\left|\expc\left[\xi_k- \expc \xi_k\right]^{n} \right|\le n!
\sigma^2_k H^{n-2},\\
\sigma_k^2:=\Var \xi_k
\end{gather*}
with $H$ being the same for all $k$. As we have mentioned in the beginning of this remark, there is no such problem for the variable in question in the current work setup. In fact the random variable in Theorem B, can be conditioned so, that it is exactly as in  "classical" Bernstein condition and $H=C\log W$ with absolute constant $C$. That is exactly the \emph{junction where the localization lengths acquires the factor $\left(\log W\right)^3$}. 

\smallskip
$(2)$ Theorem $B$ \emph{does have a version for arbitrary $N$}. Unfortunately, \emph{it does not seem possible to fit Bernstein condition for the latter seting without loosing the sharpness of the estimate}.
This objection leads to application of \emph{multi-scale analysis method} from the random Schrodinger operators, introduced in the seminal work of Frohlich and Spencer, ~\cite{FS-83-Absence}. Verification reduces to validation of what is called in multiscale analysis \emph{initial scale estimates}. Once these estimates established and the Wegner estimate is known, the technology explicitely evaluates the localization length. Theorem C below states all the estimates needed for this objective. The statement of Theorem A is exactly the result of application of multiscale analysis with estimates of Theorem C and well-known Wegner estimate in the following form
\be\label{eq:Wegner0}
\expc \left[\left\|G_{[1,N]}\left(1;N\right)\right\|^t\right]\le C_1W^{C_2}, \quad\text{ for any $0<t<{1\over 2}$}
\ee
where $C_1,C_2$ are absolute constant. \emph{We do not discuss the derivation of Theorem A from Theorem C, for it is well-known to the experts how this works}.
\end{remark}
\begin{thmc}\label{thm:C} There exists a large absolute constant $D_0\gg 1$ such that for $W^2\le N\le W^{D_0}$ the following statement holds. There exists a set $\cB_N$ such that
\[
\pr\big[\cB_N\big]\le N^{-C_0D_0},
\]
where $C_0\gg 1$ is an absolute constant, and off this set set 
\be\label{eq:LexponentTHMC}
\bn G_{[1,N]}(1;N)\bn<e^{-{cN\over W\left(\log W\right)^3}}
\ee
\end{thmc} 
\section{Variables Change.}\label{sec:VariableChange}
Take 
\begin{equation}
\label{eq:gdef}
H_{[1,N]}=\begin{bmatrix}
V_1 & T_1 & 0 & \cdots & \cdots & 0\\
T_1^t & V_2 & T_2 & 0&\cdots & 0\\
0 & T^t_2 & V_3 & T_3&0\cdots & 0\\
\cdots&\cdots&\cdots&\cdots&\cdots&\cdots\\
\cdots&\cdots&\cdots&\cdots&\cdots&\cdots\\
0 & \cdots & \cdots &0 & T^t_{N-1} & V_N\end{bmatrix},
\end{equation}
where $V_n\in \textsf M_{W,\textsf s}$, $T_n\in \textsf M_{W}$. Assume $V_n$ and $T_m$ are random matrices. Given unit vector $f\in \IR^W$ we call the map $\left((V_1,...,V_N,T_1,...T_{N-1}\right)\to G_{[1, N]}\left(1; N\right)f$ \emph{Green function edge--to--edge vector action}. The method we develop in this work targets the random variable
\be\label{eq:mainrandomvariable}
\gamma:=\log \left\|G_{[1, N]}\left(1;N\right)f\right\|
\ee
with $V_n$ being random GOE matrices, $T_n$ random from gaussian matrices ensemble, and all matrices being independent. As it was mentioned in the Introduction to study this object we employ changes of variabes. In this section we discuss these new variables and changes of variables in question

The most important change of variables targets the Green function edge--to--edge matrix factorization. The idea of effectiveness of this approach was introduced and demonstrated in the Schenker's work ~\cite{Sch-09-Eigenvector}. We call all these new variables \emph{Schenker variables}. It starts with application of Schur complement formula to the $W\times W$ diagonal blocks of $H_{[1, N]}-E$. After that it  proceeds with a change of variables involved in Schur complement formula, which allows to see "Markovian features" of the random matrix $G_{[1, N]}\left((V_1,...,V_N,T_1,\dots,T_{N-1};1;N\right)$. For convenience of the reader we state in the next proposition the details of application of Schur complement formula for non-random setting.
\begin{prop}\label{prop:Schur} Using the above notations the following equations hold 

\smallskip
$(i)$ For $n\ge 2$ 
\begin{gather*}
G_{[1, 1]}(1;1)=U_1^{-1},\quad U_1= V_1-E,\\
G_{[1,n]}(1;n) = -G_{[1, n-1]}(1; n-1)T_{n-1} U_n^{-1},\quad n\ge 2\\
U_n=\left(V_n-E-T^t_{n-1}G_{[1, n-1]}(n-1; n-1)T_{n-1}\right)
\end{gather*}

\smallskip
$(ii)$ Using the notations is $(i)$
\begin{gather*}
G_{[1,N]}(1;N) = (-1)^{N-1}U_1^{-1}\times T_1\times ...\times U^{-1}_{N-1}\times T_{N-1}\times U_N^{-1}
\end{gather*}
\end{prop}
Next we state Schenker variables basic features. The key observation by Schenker is of course that the Jacobian of the variables change is identically 1. That allows for explicit joint distribution equation. All features, including "Markovian" ones, follow straight from the definitions combined with this fact. The verification of all features is short. Since in the current setup $T_k$ are random we make detailed statements and discuss the proofs for completeness. We consider not identically distributed variables just because the features in question do not require this condition. We discuss also Schenker variables \emph{marginal and conditioned distributions}. That makes it more convenient to explain validity of certain details in the main technological estimates in Section~\ref{sec:basicestimates}. Those estimates enter \emph{spherical radius concentration and fluctuations estimates against certain super-exponential densities}, see next Section~\ref{sec:superexp}. Similar distribution features were introduced independently in the recent work of N. Chen and C.Smart ~\cite{CS22} In this section, and later in what follows, out of convenience, we suppress the condition variables from the notations of conditioned distributions, unless it causes umbiguity.

\begin{prop}\label{prop:newvariables}

\smallskip
$(i)$
Using the notations in Proposition~\ref{prop:Schur}, assume $T_k\in \textsf M_W$, $V_\ell\in \textsf M_{\textsf s,W}$, are random independent with 
\be
\quad d\pr\left(V_k\right)=\phi_k\left(V_k\right)dV_k,\quad dV_k=\prod_{1 \leq p\le q \leq W}dv_{k,p, q},\nn\\
V_k=\left(v_{k,p, q}\right)_{1 \leq p, q \leq W}, \quad v_{k,p, q}=v_{k,q,p},\qquad\nn\\
\label{eq:basicdistributions}\\
d\pr\left(T_k\right)=\rho_k\left(T_k\right)dT_k,\quad dT_k=\prod_{1 \leq p,q \leq W}d\xi_{k,p, q},\nn\\
T_k=\left(\xi_{k,p,q}\right)_{1 \leq p, q \leq W}\qquad\qquad\qquad\nn
\ee
Given  $T_1,\dots,T_{N-1}$, define new variables $u_{k,p,q}$, $1\le p\le q\le W$ in $\prod_{k=1}^{N}\IR^{{W(W+1)\over 2}}$ via
\begin{gather*}
U_1 =V_1-E,\quad U_k=\left(V_k-E-T^t_{k-1}\times U_{k-1}^{-1}\times T_{k-1}\right),\\
U_k=\left(u_{k,p, q}\right)_{1 \leq p, q \leq W}, \quad u_{k,p, q}=u_{k,q,p}
\end{gather*}
The Jacobian of the transformation 
\begin{gather*}
\IR^{{NW(W+1)\over 2}}\ni\left(v_{k,p, q}\right)_{1\le k\le N;1 \leq p\le q \leq W}\to \mathfrak F_{T_1,\dots,T_{N-1}}
\left(\left(v_{k,p, q}\right)_{1\le k\le N;1 \leq p\le q \leq W}\right):=\\
\left(u_{k,p, q}\right)_{1\le k\le N;1 \leq p\le q \leq W}\in\IR^{{NW(W+1)\over 2}}
\end{gather*}
is $=1$.

\smallskip
$(ii)$ The joint distribution of $U_1,U_2, \ldots,U_N,T_1,\dots,T_{N-1}$ is as follows
\[d\pr_N\Big(U_1,U_2, \ldots,U_N,T_1,\dots,T_{N-1}\Big) =\]
\[
 \phi_1(U_1)\times\prod_{k=2, \ldots N}\rho_{k-1}\left(T_{k-1}\right)\phi_k\left(U_k+T^t_{k-1}\times U_{k-1}^{-1}\times T_{k-1}\right)\times dU_k\times dT_{k-1}
\]

\smallskip
$(iii)$ The variables $U_1,\dots,U_{n-1},T_1,\dots,T_{n-1}$ are independent from the variables $U_{n+1},\dots,U_{N},T_{n},\dots,T_{N-1}$, when conditioned on $U_n$. 
In particular, the variables $T_{n}$, $n=1,...,N-1$ are independent, when conditioned on $U_{n}$, $n=1,...$.

\smallskip
$(iv)$  The distribution of $T_{n-1}$, conditioned on the rest of the variables is as follows 
\begin{gather*}
d\pr_N\Big(T_{n-1}\big|\Big)=\\
\mu_{T,n-1}\Big(U_{n-1},U_{n}\Big)\rho_{n-1}\left(T_{n-1}\right)\times\phi_n\left(U_n+T^t_{n-1}\times U_{n-1}^{-1}\times T_{n-1}\right)\times dT_{n-1},
\end{gather*}
where $\mu_{T,n-1}\Big(U_{n-1},U_{n}\Big)$ is the normalization factor. The joint distribution of $T_{n}$, conditioned on $U_{n}$, $n=1,...,N$ is as follows 
\begin{gather*}
d\pr_N\Big(T_1,...,T_{N-1}\big|\Big)=\\
\prod_{1 \leq n \leq N-1} \mu_{T,n}\Big(U_{n},U_{n+1}\Big)\rho_{n}\left(T_{n}\right)\times\phi_{n+1}\left(U_{n+1}+T^t_{n}\times U_{n}^{-1}\times T_{n}\right)\times dT_{n}
\end{gather*}

\smallskip
$(v)$ The $\sigma$--algebras $\cF\left(U_1,T_1,...,U_{n-1},T_{n-1}\right)$ and $\cF\left(V_n\right)$ are independent.

\end{prop}
\begin{proof} $(i)$ From the transformation definition,
\begin{gather*}
v_{k, p, q} = u_{k, p, q}+s_{k-1, p, q},\quad p\le q,\quad \big(s_{k-1, p, q}\big)_{1 \le p, q\le k} = U_{k-1}^{-1}
\end{gather*}
That implies the statement.

\smallskip
$(ii)$
Combine the change of variables distribution density rule with part $(i)$, get the statement in $(ii)$.

\smallskip
$(iii)$ Combine conditional distribution density rule with the the joint distribution equation from $(ii)$, get the statement.

\smallskip
$(iv)$ Combine the conditioned distribution rule with the joint distribution equation from part $(ii)$, get the statement.

\smallskip
$(v)$ From the definitions $\cF\left(U_1,T_1,...,U_{n-1},T_{n-1}\right)\subset\cF\left(V_1,T_1,...,V_{n-1},T_{n-1}\right)$ $($ actually those are the same $)$. Since $V_1,T_1,...,V_{n-1},T_{n-1},V_n$ are independent random variables, the statement is correct.

\end{proof}

\section{Vector Action Decoupling-Splitting. Recurrence Equation.}
\label{sec:basicestimates}
In this section we discuss the setup of the method, which is instrumental to evaluate the random variable \eqref{eq:mainrandomvariable} $($the central object of the method$)$. Recall that the random variable is as follows 
\be\label{eq:mainrandomvariablea}
\gamma:=\log \left\|G_{[1, N]}\left(1;N\right)f\right\|,
\ee
where $G_{[1, N]}\left(1;N\right)$ stands for the Green function edge--to--edge matrix and $f\in\IR^W$ is a fixed unit vector. Recall that using the notations in Proposition~\ref{prop:Schur}
\be\label{eq:Greenedgetoedge}
G_{[1,N]}(1;N)f= (-1)^{N-1}U_1^{-1}\times T_1\times ...\times U^{-1}_{N-1}\times T_{N-1}\times U_N^{-1}f
\ee
The statements in this section address the distribution of \emph{this equation factors} vector action. 
In other words, we look into the following random vector--functions
\be\nn
\textsf v_{g}\left(T_{n-1}\big|\right):=T_{n-1}\times U_n^{-1}g,\qquad\qquad\qquad\qquad\qquad\\
\label{eq:schurfactorvectoraction}\\
\qquad\text{with variables $U_n,U_{n-1}$, being fixed and $g$, being an arbitrary fixed unit vector}\nn
\ee
We setup a technology which allows for a relay from a given factor to the next one on its left, so that\emph{ this one also acts on a fixed unit vector}, i.e. vector action starts over with a completely similar setup. The technology also \emph{"decouples" the vector action outcome from the rest of components of the Schenker variable in  question}. In other words it provides an explicit equation for the vector action outcome distribution conditioned on the rest of the components.

We do not discuss in this section how to combine the estimates with $n=1,2,...$, we just look into each given factor action only. Moreover, here we do not address the most important question concerning the random vector--functions
in \eqref{eq:schurfactorvectoraction},-- \emph{what is the magnitue of fluctuations, of $\log\|\textsf v_n\|$, conditioned on 
\be\label{eq:vectorcond}
\left\|\textsf v_n\left(T_{n-1}\big|\right)\right\|^{-1}\textsf v_n\left(T_{n-1}\big|\right)=h
\ee
with a given unit vector  $h\in\IR^W$ $($ the one to relay to the neighbour factor$)$} We address these questions in Section~\ref{sec:Theorem B}. It turns out that the magnitude of the fluctuations in question is \emph{completely defined by magnitude of either of two specific vectors in the splitting setup}. We call these vectors \emph{splitting principal vectors} and denote them $\textsf B,\textsf Z$. We will show in Section~\ref{sec:Theorem B} that  \emph{the magnitude of the fluctuations in question is of order $\sim W^{-{1\over 2}}\left\|\textsf Z\right\|^{-1}\sim W^{-{1\over 2}}\left\|\textsf B\right\|^{-1}$}. 

The principal vector $\textsf Z$ naturally indexed by $n-1$ from \eqref{eq:vectorcond}. The number $\left\|\textsf Z\right\|$ itself is \emph{of magnitude of the Schur complement factor vector action $U_{n-2}^{-1}h$}, see \eqref{eq:Greenedgetoedge} and \eqref{eq:vectorcond}.
In the current section we \emph{identify a recurrence equation relating the vectors $\textsf Z_n$ and  $\textsf B_{n+1}$}. This is the key technological piece. In Section~\ref{sec:mainestimates} we show that with high probability
\be\label{eq:recurrence0}
\min\left(\left\|\textsf Z_{n}\right\|,\left\|\textsf B_{n+1}\right\|\right)\sim 1
\ee
This equation in the display does not hold eveywhere,- it is much more subtle. The definition of the objects needed to state it properly is pretty lengthy. This is the crucial technological estimate because it implies that at least for every second $n$ the magnitude of the fluctuations in question is $\sim W^{-{1\over 2}}$ exactly. The details missing in the display equation are delicate. To incorporate these details a long series of involved estimates is needed. In Remark~\ref{rem:recurrence} we comment on \emph{elementary nature of the recurrence equation itself, but still being effective}.

We state the main results of the current section in Proposition~\ref{prop:condecoupling1} below. Recurrence equation appears in part $(4)$ of the proposition. We assume that the random matrices $V_k$ are GOE random matrices, i.e. $\phi_k=\phi$ as in \eqref{eq:GOEmatrix}. The random matrices $T_k$ are Gaussian non-symmetric, i.e. $\rho_k=\rho$, as in \eqref{eq:Gaussianonsym}.

\begin{prop}\label{prop:condecoupling1} Use the notations in Proposition~\ref{prop:newvariables} with $\phi_k=\phi$, $\rho_k=\rho$.

\smallskip
$(1)$ Use the above notations in this section. Take unit vector $g\in\IR^W$. Take arbitrary orthogonal matrix $Q_1=Q_1\left(U_n,g\right)\in\mathbb O_W$, such that
\be\label{eq:vectorchange0}
\left\|U_n^{-1}g\right\|^{-1}Q_1U_n^{-1}g=e_1,
\ee
where $e_1$ is the first standard basis in $\IR^W$ vector. Change the variable $T_{n-1}=R_{n-1}Q_1$. With this change of variable the distribution of the random vector $\textsf v_{g}\left(T_{n-1}\big |\right)$ is the the push--forward of the following distribution
\begin{gather*}
d\pr\left(R_{n-1}\big|\right)=
\kappa
e^{-{W\over 2}Tr \left[R_{n-1}^tR_{n-1}\right]}\times e^{-{W\over 4}Tr \left[ U_n+Q_1^tR_{n-1}^t\times U_{n-1}^{-1}\times R_{n-1}Q_1\right]^2}dR_{n-1}
\end{gather*}
under the map 
\be\label{eq:vectorchange1}
R_{n-1}\to\textsf v\left(R_{n-1}\big|\right):=\left\|U_n^{-1}g\right\| R_{n-1}e_1,
\ee 
Here $\kappa=\kappa\left(U_{n-1},U_n\right)$ is the normalization factor.

\smallskip
$(2)$ Use notations in part $(1)$. There exists orthogonal matrix $Q_2=Q_2\left(U_{n-1}\right)\in\mathbb O_W$, such that with new variable $$S_{n-1}:=S=Q^t_2R_{n-1}=Q^t_2T_{n-1}Q^t_1,$$ 
holds 
\be\label{eq:vectorchange2}
\textsf v\left(R_{n-1}\big|\right)=\left\|U_n^{-1}g\right\|Q_2Se_1,
\ee
and
\begin{gather*}
d\pr\left(S\big|\right)=\kappa
e^{-{W\over 2}Tr \left[S^tS\right]}\times e^{-{W\over 4}Tr \left[ Q_1U_nQ_1^t+S^t\times C\times S\right]^2}\times dS\times dU_n,\\
Tr \left[ Q_1U_nQ_1^t+S^t\times C\times S\right]^2=Tr \left[\lambda_1\tilde Y^t_1\times\tilde Y_1+\check{S}^t\check C\check{S}+\check B \right]^2\\
+\left(\lambda_{1}\xi^2_{1,1}+\langle \check C\tilde X_1,\tilde X_1\rangle+b_{1,1}\right)^2+2\left\|\lambda_{1}\xi_{1,1}\tilde Y^t_1+\check{S}^t\check C\tilde X_1+\tilde B_1\right\|^2,
\end{gather*}
where 
\begin{gather*}
S=\left(\xi_{i,j}\right)_{1 \leq i,j \leq W},\quad \check S=\left(\xi_{i,j}\right)_{2\le i,j\le W},\\
\tilde Y_1=\left(\xi_{1,j}\right)_{2\le j\le W},\quad \tilde X_1=\left(\xi_{i,1}\right)^t_{2\le i\le W},\\
Q_1U_nQ_1^t:=B:=\left(b_{i,j}\right)_{1\leq i,j\leq W},\quad B_1=Be_1=\left(b_{i,1}\right)^t_{1\le i\le W},\\
\check B=\left(b_{i,j}\right)_{2\le i,j\le W},\quad \tilde B_1=\left(b_{i,1}\right)^t_{2\le i\le W},\\
C:=Q^t_2U^{-1}_{n-1}Q_2,\\
C=\left(\lambda_i\delta_{i,j}\right)_{1\leq i,j\leq W},\quad \check C=\left(\lambda_i\delta_{i,j}\right)_{2\le i,j\le W},
\end{gather*}
Here $\left\|X\right\|$ stands for the standard Euclidean norm of the vector. Furthermore, with $X_1^t:=\left(\xi_{1,1},\tilde X^t_1\right)$ holds
\begin{gather*}
S^tCX_1=
\begin{bmatrix}
\lambda_1\xi^2_{1,1}+\langle \check C\tilde X_1,\tilde X_1\rangle\\
\lambda_1\xi_{1,1}\tilde Y^t_1+\check{S}^t\check C\tilde X_1\end{bmatrix},
\end{gather*}

\smallskip
$(3)$ Use the notations in part $(2)$. The conditional distribution of the random vector $X_1^t$ is as follows
\begin{gather*}
d\pr\left(\left(\xi_{1,1},\tilde X^t_1\right)^t\big|\right)=\mu\Big(U_{n-1},U_{n}\Big)e^{-{W\over 2}\phi\left( X_1\big|\right)}\times dX_1,\\
\phi\left(X_1\big|\right)=\left\|X_1\right\|^2+{1\over 2}\left(\langle C X_1,X_1\rangle+b_{1,1}\right)^2+\\
\left\|\lambda_{1}\xi_{1,1}\tilde Y^t_1+\check{S}^t\check C\tilde X_1+\tilde B_1\right\|^2
\end{gather*}
Here $\mu\Big(U_{n-1},U_{n}\Big)$ is the normalization factor. The vector function $\textsf v\left(T_{n-1}\big|\right)$ from \eqref{eq:schurfactorvectoraction} and the conditioning equation \eqref{eq:vectorcond} are as follows
\be\nn
\textsf v\left(T_{n-1}\big|\right)=Q_2\left(\xi_{1,1},\tilde X^t_1\right)^t,\qquad\\
\label{eq:vectorcond1}\\
\quad\left\|\left(\xi_{1,1},\tilde X^t_1\right)\right\|^{-1}Q_2\left(\xi_{1,1},\tilde X^t_1\right)^t=h\nn
\ee

\smallskip
$(4)$  Use the notations in part $(2)$, attach the index $(n-1)$ to $S$ back i.e. use $S_{n-1}=Q^t_2R_{n-1}$. Add the index $(n-1)$ to some of the notations in parts $(2)$,  $(3)$:
\be
\label{eq:siteindex}B_{1,n-1}:=B_1,\quad C_{n-1}:=C,\quad X_{1,n-1}^t:=\left(\xi_{1,1,n-1},\tilde X^t_{1,n-1}\right) 
\ee 
Introduce the site $(n-1)$--\textsf{principal vectors}:
\be
\label{eq:princevector} \textsf B_{n-1}=B_{1,n-1},\quad \textsf Z_{n-1}=C_{n-1}X_{1,n-1}
\ee
Define all these objects for the site $(n-2)$ with $h$ as in part $(3)$ in the role of $g$. Then the following \textsf{recurrence equation} holds
\be
\label{eq:recurrence} \left\|\textsf B_{n-2}\right\|=\left\|X_{1,n-1}\right\|\left\|\textsf Z_{n-1}\right\|^{-1}
\ee

\end{prop}
\begin{remark}\label{rem:recurrence} 
It is completely clear, from common sense point, that the recurrence equation \eqref{eq:recurrence} must be nothing, but a simple computation, which follows straightforward from the the Green function edge--to--edge matrix structure in \eqref{eq:Greenedgetoedge}. That is true, but is hard see it in the original setup. The main point is that it is not so clear why this helps to evaluate the fluctuations in question. It becomes transparent in the splitting-decoupling setup development though.
\end{remark}
Before we prove of Proposition~\ref{prop:condecoupling1} we want to recall one equation from Proposition~\ref{prop:newvariables}, to which we refer repeatedly in this section and the next one 
\be\nn
d\pr_N\Big(T_{n-1}\big|\Big)=\mu_{T,n-1}\Big(U_{n-1},U_{n}\Big)\times\qquad\qquad\\
\label{eq:tripletdistributn}\\
\nn\rho_{n-1}\left(T_{n-1}\right)\times\phi_n\left(U_n+T^t_{n-1}\times U_{n-1}^{-1}\times T_{n-1}\right)\times dT_{n-1}
\ee
where $\mu_{T,n-1}\Big(U_{n-1},U_{n}\Big)$ is the normalization factor.
We also need the following elementary linear algebra statement.
\begin{lemma}\label{lem:elem1} Take arbitrary matrices $A\in \textsf M_{W,\textsf s}$ and $S=\left(\xi_{i,j}\right)_{1\le i,j\le W}\in \textsf M_{W}$. Then there exists orthogonal matrix $Q\in\mathbb O_W$, such that 
\begin{gather*}
S^tQ^t\times A\times QS=
\begin{bmatrix}
\lambda_1\xi^2_{1,1}+\langle\check C\tilde X_1,\tilde X_1\rangle&\lambda_1\xi_{1,1}\tilde Y_1+\left(\check{S}^t\check C\tilde X_1\right)^t\\
\lambda_1\xi_{1,1}\tilde Y^t_1+\check{S}^t\check C\tilde X_1&\lambda_1\tilde Y^t_1\times\tilde Y_1+ \check{S}^t\check C\check{S}\end{bmatrix}
\end{gather*}
where 
$\check S\in \textsf M_{W-1}$, and the matrices $\check S$, $\tilde X_1$, $\tilde Y_1$ come from the following block-matrix representation 
\begin{gather*}
S=
\begin{bmatrix}
\xi_{1,1}&\tilde Y_1\\
\tilde X_1&\check{S}\end{bmatrix},
\end{gather*}
$\lambda_1,\lambda_2,\dots,\lambda_W$ are arbitrary enumerated eigenvalues of $A$, $\check C=\left(\delta_{i,j}\lambda_j\right)_{2\le i,j\le W}$. Furthermore, with $X_1^t:=\left(\xi_{1,1},\tilde X^t_1\right)$ holds
\begin{gather*}
S^tCX_1=
\begin{bmatrix}
\lambda_1\xi^2_{1,1}+\langle \check C\tilde X_1,\tilde X_1\rangle\\
\lambda_1\xi_{1,1}\tilde Y^t_1+\check{S}^t\check C\tilde X_1\end{bmatrix},
\end{gather*}

\end{lemma}
\begin{proof} Find orthogonal matrix $Q\in\mathbb O_W$ such that $Q^t\times A\times Q$ is a diagonal matrix with $\lambda_1,\dots,\lambda_W$ on its diagonal. Use block matrices mutiplication rule, get
\begin{gather*}
S^t\times Q^t\times A\times Q\times S=\begin{bmatrix}
\xi_{1,1}&\tilde X^t_1\\
\tilde Y^t_1&\check{S}^t\end{bmatrix}\times \begin{bmatrix}
\lambda_1&0\\
0&\check C\end{bmatrix}\times \begin{bmatrix}
\xi_{1,1}&\tilde Y_1\\
\tilde X_1&\check{S}\end{bmatrix}=\\
\begin{bmatrix}
\xi_{1,1}&\tilde X^t_1\\
\tilde Y^t_1&\check{S}^t\end{bmatrix}\times\begin{bmatrix}
\lambda_1\xi_{1,1}&\lambda_1\tilde Y_1\\
\check C\tilde X_1&\check C\check{S}\end{bmatrix}=
\begin{bmatrix}
\lambda_1\xi^2_{1,1}+\langle \check C\tilde X_1,\tilde X_1\rangle&\lambda_1\xi_{1,1}\tilde Y_1+\left(\check{S}^t\check C\tilde X_1\right)^t\\
\lambda_1\xi_{1,1}\tilde Y^t_1+\check{S}^t\check C\tilde X_1&\lambda_1\tilde Y^t_1\times\tilde Y_1+ \check{S}^t\check C\check{S}\end{bmatrix},
\end{gather*}
as claimed. Take $X_1^t:=\left(\xi_{1,1},\tilde X^t_1\right)$,  get 
\begin{gather*}
S^tCX_1=
\begin{bmatrix}
\xi_{1,1}&\tilde X^t_1\\
\tilde Y^t_1&\check{R}^t\end{bmatrix}\times \begin{bmatrix}
\lambda_1\xi_{1,1}\\
\check C\tilde X_1\end{bmatrix}=
\begin{bmatrix}
\lambda_1\xi^2_{1,1}+\langle \check C\tilde X_1,\tilde X_1\rangle\\
\lambda_1\xi_{1,1}\tilde Y^t_1+\check{R}^t\check C\tilde X_1\end{bmatrix},
\end{gather*}
as claimed.
\end{proof}

\begin{proof}[Proof of Proposition ~\ref{prop:condecoupling1}] $(1)$  Use the proposition notations. Take unit vector $g\in\IR^W$.  
Take arbitrary orthogonal matrix $Q_1=Q_1\left(U_n,g\right)\in\mathbb O_W$, such that
\be\nn
\left\|U_n^{-1}g\right\|^{-1}Q_1U_n^{-1}g=e_1,
\ee
where $e_1$ is the first standard basis in $\IR^W$ vector. Change the variable $T_{n-1}=S_{n-1}Q_1$ with the rest of the variables being intact. Then
\be\nn
\textsf v_{g}\left(T_{n-1}\big|\right)=T_{n-1}\times U_n^{-1}g=R_{n-1}Q\times U_n^{-1}g=\\
\label{eq:vectorchange}\\
\nn\left\|U_n^{-1}g\right\|^{-1}S_{n-1}e_1=
\textsf v\left(S_{n-1}\big|\right),\qquad\qquad
\ee
as in \eqref{eq:vectorchange1} in part $(1)$, see the definition equation \eqref{eq:schurfactorvectoraction}.
The Jacobian of the variable change is $=1$. Combine the change of variables distribution density rule with Jaconian $=1$, with the second equation in \eqref{eq:tripletdistributn}, and with invariance of $\rho_{n-1}$ under the change of variable, get the statement in $(1)$.

\smallskip
$(2)$ Apply Lemma~\ref{lem:elem1} with $S_{n-1}$ in the role of $T$ and $U_{n-1}^{-1}$ in the role of $A$, find $Q:=Q_2=Q_2\left(U_{n-1}\right)\in\mathbb O_W$ as stated in the lemma. Change the variable $S=Q^t_2S_{n-1}$ as stated in $(2)$, write
\begin{gather*}
S_{n-1}^t\times U_{n-1}^{-1}\times S_{n-1}=S^t\times Q_2^tU_{n-1}^{-1}Q_2\times S=\\
\begin{bmatrix}
\lambda_1\xi^2_{1,1}+\langle \check C\tilde X_1,\tilde X_1\rangle&\lambda_1\xi_{1,1}\tilde Y_1+\left(\check{T}^t\check C\tilde X_1\right)^t\\
\lambda_1\xi_{1,1}\tilde Y^t_1+\check{T}^t\check C\tilde X_1&\lambda_1\tilde Y^t_1\times\tilde Y_1+ \check{T}^t\check C\check{T}\end{bmatrix},\\
S^tCX_1=
\begin{bmatrix}
\lambda_1\xi^2_{1,1}+\langle \check C\tilde X_1,\tilde X_1\rangle\\
\lambda_1\xi_{1,1}\tilde Y^t_1+\check{R}^t\check C\tilde X_1\end{bmatrix},
\end{gather*}
see the notations in the statement of part $(2)$. Change the variable $S_{n-1}=Q_2S$ in the equation \eqref{eq:tripletdistributn}, get
\be
\label{eq:vectorchangeS}
\textsf v_{g}\left(S_{n-1}\big|\right)=\left\|U_n^{-1}g\right\|^{-1}Q_2Se_1,
\ee
as in \eqref{eq:vectorchange2} in part $(2)$. 
The Jacobian of the variable change is $=1$.
Use the distribution equation from part $(1)$, get
\begin{gather*}
d\pr_N\Big(S\big|\Big)=
\mu\Big(U_{n-1},U_{n}\Big)\times e^{-{W\over 2}Tr S^tS}\times e^{-{W\over 4}Tr \left[U_n+Q_1^tS^tQ_2^t\times U_{n-1}^{-1}\times Q_2SQ_1\right]^2}dS
\end{gather*}
Use that and
\begin{gather*}
Tr\left[Q^t\times X\times Q\right]^2=Tr\left[X^2\right],\quad Q^tQ=I
\end{gather*}
for any $X\in \textsf M_{W,\textsf s}$ and $Q\in\mathbb O_W$, rewrite
\begin{gather*}
Tr\left[U_n+Q_1^tS^tQ_2^t\times U_{n-1}^{-1}\times Q_2SQ_1\right]^2=Tr\left[Q_1U_nQ_1^t+S^tQ_2^t\times U_{n-1}^{-1}\times Q_2S\right]^2,\\
d\pr_N\Big(S\big|\Big)=
\hat \mu\Big(U_{n-1},U_{n}\Big)\times e^{-{W\over 2}Tr S^tS}\times e^{-{W\over 4}Tr\left[Q_1U_nQ_1^t+S^tQ_2^t\times U_{n-1}^{-1}\times Q_2S\right]^2}dS
\end{gather*}
Combine, use the notations in part $(2)$ statement, get 
\be\nn
Q_1U_nQ_1^t+S_{n-1}^t\times U_{n-1}^{-1}\times S_{n-1}=B+S_{n-1}^t\times U_{n-1}^{-1}\times S_{n-1}=\\
\label{eq:mainsplit}\\
\nn\left(b_{i,j}\right)_{1\leq i,j\leq W}+\begin{bmatrix}
\lambda_1\xi^2_{1,1}+\langle \check C\tilde X_1,\tilde X_1\rangle&\lambda_1\xi_{1,1}\tilde Y_1+\left(\check{T}^t\check C\tilde X_1\right)^t\\
\lambda_1\xi_{1,1}\tilde Y^t_1+\check{T}^t\check C\tilde X_1&\lambda_1\tilde Y^t_1\times\tilde Y_1+ \check{T}^t\check C\check{T}\end{bmatrix}
\ee
Split the matrix $B$ in blocks, use 
\begin{gather*}
Tr\begin{bmatrix}
x&\check X^t_1\\
\check X_1&\check{X}\end{bmatrix}^2=x^2+2\left\|\check X_1\right\|^2+Tr\left[\check X^2\right]
\end{gather*}
for any symmetric block matrix, get the distribution statement in part $(2)$.

\smallskip
$(3)$ Use the notations in part $(2)$. Use conditional distribution rule, get the first statement in $(3)$. Combine the definition equation \eqref{eq:schurfactorvectoraction}, respectively, the conditioning equation \eqref{eq:vectorcond}, with equations \eqref{eq:vectorchange} and \eqref{eq:vectorchangeS}, get the second statement in $(3)$

\smallskip
$(4)$ We derive an equation for the site $(n-1)$, which clearly projects onto a statement for the site $(n-2)$, and which implies easily the statement in $(4)$. This way we do not need to introduce additional notations. Use equation \eqref{eq:vectorchange0}, combine with equations in parts $(2)$, write
\begin{gather*}
\left\|U_n^{-1}g\right\|^{-1}Q_1U_n^{-1}g=e_1,\\
\left\|U_n^{-1}g\right\|^{-1}U_n^{-1}g=Q_1^te_1,\\
\textsf B_{n-1}=B_{1,n-1}=Be_1=Q_1U_nQ_1^te_1=\left\|U_n^{-1}g\right\|^{-1}Q_1g,\\
\left\|\textsf B_{n-1}\right\|=\left\|U_n^{-1}g\right\|^{-1}
\end{gather*}
The last equation here clearly projects on the  $(n-2)$ site:
\begin{gather*}
\left\|\textsf B_{n-2}\right\|=\left\|U_{n-1}^{-1}h\right\|^{-1}
\end{gather*}
Recall that $h$ here is as in part $(3)$, i.e. 
\begin{gather*}
h=\left\|\left(\xi_{1,1},\tilde X^t_1\right)\right\|^{-1}Q_2\left(\xi_{1,1},\tilde X^t_1\right)^t=\left\|X_{1,n-1}\right\|^{-1}Q_2X_{1,n-1}
\end{gather*}
Write
\begin{gather*}
\left\|U_{n-1}^{-1}h\right\|=\left\|X_{1,n-1}\right\|^{-1}\left\|U_{n-1}^{-1}Q_2X_{1,n-1}\right\|=\left\|X_{1,n-1}\right\|^{-1}\left\|Q_2^tU_{n-1}^{-1}Q_2X_{1,n-1}\right\|=\\
\left\|X_{1,n-1}\right\|^{-1}\left\|C_{n-1}X_{1,n-1}\right\|=\left\|X_{1,n-1}\right\|^{-1}\left\|\text Z_{n-1}\right\|
\end{gather*}
Combine, get the recurrence equation in the statement.

\end{proof}

\section{Main Estimates.}
\label{sec:mainestimates}

We state the main results of the current section in Proposition~\ref{prop:condecoupling} below.

\begin{prop}\label{prop:condecoupling} Use the notations in Proposition~\ref{prop:condecoupling1}.

\smallskip
$(0)$  Use the notations in Proposition~\ref{prop:condecoupling1}, part $(3)$. Introduce the spherical variables 
\begin{gather*}
\left(\xi_{1,1},\tilde X^t_1\right)^t:=X=r \textsf x,\\
r=\left\|X\right\|,\quad \textsf x=\left(\textsf x_{1},...,\textsf x_{W}\right)=r^{-1}X^t\in\mathbb S^{W-1}
\end{gather*}
The joint conditional distribution $d\pr\left(r,\textsf x\big|\right)$ and the conditional distribution $d\pr\left(r\big|\right)$ are, respectively, as follows
\begin{gather*}
d\pr\left(r,x\big|\right)=\mu\Big(U_{n-1},U_{n}\Big)r^{W-1}e^{-{W\over 2}\phi\left(r,\textsf x|\right)}drd\textsf x,\\
d\pr\left(r\big|\right)=\mu\Big(U_{n-1},U_{n},\textsf x\Big)r^{W-1}e^{-{W\over 2}\phi\left(r|\right)}dr,\\
\phi\left(r,\textsf x|\right)=\phi\left(r\textsf x|\right),\quad\text{see Proposition~\ref{prop:condecoupling1}, part $(3)$},\\
\phi\left(r|\right)=\textsf a^2 r^4+\textsf b r^2+\textsf c r,\\
\textsf a^2=\textsf a^2\left(\textsf x\right)={1\over 2}\langle C\textsf x,\textsf x\rangle^2,\\
\textsf b=\textsf b\left(\textsf x\right)=1+\beta^2+\gamma,\\
\beta^2=\beta^2\left(\textsf x\right)=\left\|\lambda_{1}\textsf x_{1}\tilde Y^t_1+\check{S}^t\check C\tilde{\textsf x}\right\|^2,\quad \tilde{\textsf x}=\left(\textsf x_{2},...,\textsf x_{W}\right)^t,\\
\gamma=\gamma\left(\textsf x\right)=b_{1,1}\langle C\textsf x,\textsf x\rangle,\\
\textsf c=\textsf c\left(\textsf x\right)=
2\langle\lambda_{1}\textsf x_{1}\tilde Y^t_1+\check{S}^t\check C\tilde{\textsf x},\tilde B_1\rangle,
\end{gather*}
Here $\mu\Big(U_{n-1},U_{n}\Big)$, $\mu\Big(U_{n-1},U_{n},\textsf x\Big)$ are the normalizing factors.

\smallskip   
$(1)$  There exists a set $\widehat{\cB}_{U,n-2,n-1,D}\in \mathcal F\left(U_{n-2},U_{n-1}\right)$,  such that
\begin{gather*}
\pr_N\left[\widehat{\cB}_{U,n-1,n,D,A}\right]\le CW^{D-1},
\end{gather*}
and a set $\cB_{S,n-1,A}\in \mathcal F\left(S\right)$ with 
\begin{gather*}
\pr_N\cB_{S,n-1,A}\le C_1e^{-c_1A^2W}
\end{gather*}
such that for any $\left(U_{n-1},U_{n}\right)\notin \widehat{\cB}_{U,n-1,n,D,A}$, $S\notin \cB_{S,n-1,A}$ and any $\textsf x$ holds 
\begin{gather*}
\max\left(\textsf a^2\left(\textsf x\right),\left|\textsf b\left(\textsf x\right)\right|,\left|\textsf c\left(\textsf x\right)\right|\right)\le 2A^2W^{2D}
\end{gather*}
Furthermore, 
\begin{gather*}
\pr_N\left\{\textsf a^2\left(\textsf x\right)=0\right\}=0
\end{gather*}

\smallskip
$(2)$ 
There exists a set $\widehat{\cB}_{U,n-1,n,A}\in \mathcal F\left(U_{n-1},U_{n}\right)$, such that
\begin{gather*}
\pr_N\left[\widehat{\cB}_{U,n-1,n,A}\right]\le Ce^{-cA^2W},
\end{gather*}
with absolute constants $c,C$, and for any $\left(U_{n-1},U_{n}\right)\notin \widehat{\cB}_{U,n-1,n,A}$ there exists a set $\cB_{X,n-1,n,A}\in \mathcal F\left(r,\textsf x\right)$, such that
\begin{gather*}
\pr_N\left[\cB_{X,n-1,n,A}\big|U_{n-1},U_{n}\right]\le Ce^{-cA^2W},
\end{gather*}
such that for any $\left(r,\textsf x\right)\notin \cB_{X,n-1,n,A}$ the following equations hold
\begin{gather*}
r\le A,\\
\bm\gamma\left(\textsf x\right)r^2+2\textsf a^2\left(\textsf x\right)r^4\bm\le \sqrt 2A\textsf a\left(\textsf x\right)r^2,\\
\left|\textsf c\left(\textsf x\right)r+2\beta^2\left(\textsf x\right)r^2\right|\le 2A\beta\left(\textsf x\right)r,
\end{gather*}
and also
\begin{gather*}
\textsf ar^2< A+\left\|B_{1}\right\|,\\
|\gamma|r^2\le \left\|B_{1}\right\|\left(A+\left\|B_{1}\right\|\right),\\
\beta^2r^2\le \left( A+\left\|B_{1}\right\|\right)^2,\\
|b|r^2\le A^2+\left( A+\left\|B_{1}\right\|\right)^2+\left\|B_{1}\right\|\left(A+\left\|B_{1}\right\|\right),\\
|\textsf c|r\le 2 \left\|B_{1}\right\|\left(A+\left\|B_{1}\right\|\right)
\end{gather*}

\smallskip
$(3)$ Use the notations in $(2)$. For any $\left(U_{n-1},U_{n}\right)\notin \widehat{\cB}_{U,n-1,n,A}$ there exists a set $\cB_{\texttt x,n-1,n,A}\in \mathcal F\left(\textsf x\right)$, such that
\begin{gather*}
\pr_N\left[\cB_{\texttt x,n-1,n,A}\big|U_{n-1},U_{n}\right]\le C_0e^{-c_0A^2W},
\end{gather*}
with absolute constants $c_0,C_0$, and for any $\textsf x\notin \widehat{\cB}_{\texttt x,n-1,n,A}$ there exists a set $\cB_{\texttt r,\textsf x,n-1,n,A}\subset (0,+\infty)$, such that
\begin{gather*}
\pr_N\left[\left\{r\in\cB_{\texttt r,n-1,n,A}\right\}\big|U_{n-1},U_{n}\right]\le C_0e^{-cA^2W},
\end{gather*}
such that for any $r\notin \cB_{\texttt r,n-1,n,A}$ the equations in part $(2)$ hold.

\smallskip
$(4)$  Take $\textsf x\notin \widehat{\cB}_{\texttt x,n-1,n,A}$, see $(3)$ \underline{with no condition $r\notin \cB_{\texttt r,\textsf x,n-1,n,A}$ though}. Take one of the $(n-1)$--principal vectors $\textsf Z=\textsf Z_{n-1}(\textsf x)=CX$, defined in \eqref{eq:princevector}, set $\zeta=\left\|\textsf X\right\|^{-1}\left\|\textsf Z\right\|$.
Then
\begin{gather*}
\textsf a^2\le {1\over 2}\zeta^2,\quad \beta^2\le A^2\zeta^2,\\
\left|\textsf b\right|\le 1+A\zeta+2A^2\zeta^2,\\
\left|\textsf c\right|\le 2A^3\zeta^2+2A\zeta
\end{gather*}

\smallskip
$(5)$ 
In the above notation attach the index $(n-1)$ to 
$$\textsf a,\textsf b, \textsf c,\beta,\gamma,\zeta,B_1,$$ i.e. use 
$$\textsf a_{n-1},\textsf b_{n-1}, \textsf c_{n-1},\beta_{n-1},\gamma_{n-1},\zeta_{n-1},B_{1,n-1}$$ 
instead. Define 
$$\textsf a_{n-2},\textsf b_{n-2}, \textsf c_{n-2},\beta_{n-2},\gamma_{n-2},\zeta_{n-2},B_{1,n-2}$$ 
in a completely similar way. 
There exists an absolute constant $\textsf C$, such that the following statement holds. Take 
\begin{gather*}
\textsf y\notin \widehat{\cB}_{\texttt x,n-2,n-1,A},\quad \rho\notin \cB_{\texttt r,\textsf x,n-2,n-1,A}
\end{gather*}
If 
$$\zeta_{n-1}\left(\textsf x\right)>\textsf C,$$
then the following equations hold
\begin{gather*}
\textsf a_{n-2}\left(\textsf y\right)\rho^2< A+\textsf C^{-1},\\
|\gamma_{n-2}\left(\textsf y\right)|\rho^2\le \textsf C^{-1}\left(A+\textsf C^{-1}\right),\\
\beta^2_{n-2}\left(\textsf y\right)\rho^2\le \left( A+\textsf C^{-1}\right)^2,\\
|\textsf c_{n-2}\left(\textsf y\right)|\rho\le 2 \textsf C^{-1}\left(A+\textsf C^{-1}\right)
\end{gather*}
Furthermore, either
\be\label{eq:bcrucial}
\textsf b_{n-2}\left(\textsf y\right)>0, 
\ee
or 
\be\label{eq:bcruciald}
|\textsf b_{n-2}\left(\textsf y\right)|<2\sqrt 2\textsf C^{-1}|\textsf a_{n-2}\left(\textsf y\right)|
\ee

In other words, using the above notations, the following \underline{dichotomy holds}: There exists an absolute constant $\textsf A$, such that for any $n$ either
\be
\label{eq:dichotomy0}\max\left(\textsf a_{n-1}^2(\textsf x),\left|\textsf b_{n-1}(\textsf x)\right|,\left|\textsf c_{n-1}(\textsf x)\right|\right)\le \textsf A
\ee
for any $\textsf x\notin \widehat{\cB}_{\texttt x,n-1,n,A}$, or
\be
\nn\max\left(\textsf a_{n-2}^2(\textsf y)\rho^4,\left|\textsf b_{n-2}(\textsf y)\right|\rho^2,\left|\textsf c_{n-2}(\textsf x)\right|\rho\right)\le \textsf A
\label{eq:dichotomy}
\ee
for any $\textsf y\notin \widehat{\cB}_{\texttt x,n-2,n-1,A}$, $\rho\notin \cB_{\texttt r,\textsf x,n-2,n-1,A}$. Furthermore, in this case, for any $\textsf y\notin \widehat{\cB}_{\texttt x,n-2,n-1,A}$, either
\be\label{eq:bcruciala}
\textsf b_{n-2}\left(\textsf y\right)>0, 
\ee
or 
\be\label{eq:bcrucialda}
|\textsf b_{n-2}\left(\textsf y\right)|<2\sqrt 2\textsf C^{-1}|\textsf a_{n-2}\left(\textsf y\right)|
\ee

\end{prop}
\begin{remark}\label{rem:sphericaldistribution} 
$(1)$ We want to remark here on equation \eqref{eq:bcrucial}. This one is crucial for the dichotomy second case. It is a key to the log-variance of spherical radius against special super-exponential densitities estimate, which we establish in section \ref{sec:superexp}. The latter, in its turn, is one of the main ingredients in the proof of Theorem B. The reader will note that in the proof of the log-variance estimate in section \ref{sec:superexp}, we use other estimates in the display above \eqref{eq:bcrucial}. This is confusing. The delicate detail, which explains the confusion is as follows. The estimates in the display hold \emph{only under condition $\rho\notin \cB_{\texttt r,\textsf x,n-2,n-1,A}$}. The set $\cB_{\texttt r,\textsf x,n-2,n-1,A}$ is negligible, for its probability is very small. However, one needs to know \emph{that the estimates hold for the value $\rho_1$, at which the exponent in the density in question assumes its maximum}. It is very natural to expect that the bulk of the distribution sits around $\rho_1$ and most of the points there must be out of the set $\cB_{\texttt r,\textsf x,n-2,n-1,A}$. This is a delicate problem though. The equation \eqref{eq:bcrucial} resolves this problem by making all terms of a central piece inequality relating \emph{the quantities in question positive, which implies that each term obeys the inequality}, i.e. there is no cancelations hidden.

\smallskip
$(2)$
Write the vector action from Proposition~\ref{prop:condecoupling1} in spherical coordinates $\textsf v\left(T_{n-1}\big|\right)=r \textsf y$. As we had mentioned in the Introduction the spherical variable $\textsf y\in S^{W-1}$ distribution is hard to evaluate. That is the reason there is no statement addressing this distribution in Proposition~\ref{prop:condecoupling}. 
\end{remark}

The proof of Proposition~\ref{prop:condecoupling} uses two statements which we discuss separately. 
We use extensively Proposition~\ref{prop:condecoupling1} with the same notations. Out of convenience, we write up below the equation we mainly use 
\be\nn
\nn d\pr_N\Big(T_{n-1}\big|\Big) =\qquad\qquad\qquad\qquad\\
\label{eq:tripletdistributn1}\\
\nn\qquad\qquad \rho_{n-1}\left(T_{n-1}\right)\times\phi_n\left(U_n+T^t_{n-1}\times U_{n-1}^{-1}\times T_{n-1}\right)\times dT_{n-1}\times dU_n
\ee

\smallskip
$\blacktriangle$ We use Gaussian matrix operator norm upper tail estimate, which must be well-known. We follow the method in Corollary 2.3.5, from the monohraph ~\cite{TAO-13}. The statement requires $M=\left(\xi_{p, q}\right)_{1 \leq p, q \leq W}$ with \emph{random variables $\xi_{p, q}$ being i.i.d. and uniformly bounded by one}. The derivation goes via \emph{$\epsilon$--net argument} introduced in the proof of Corollary 2.3.5 and based on the following fixed vector estimate from ~\cite{TAO-13}, Lemma 2.3.1: 
\begin{gather*}
\pr\left\{\left\|Mx\right\|>A\sqrt W\right\}
\le Ce^{-cAW},
\end{gather*}
for any unit vector $x\in\IR^W$, with absolute constants $c,C$.  The statement then follows from unit sphere $\mathbb S^{W-1}$ $\epsilon$-net cardinality estimate. Nothing else is needed for the $\epsilon$--net argument. For Gaussian non-symmetric matrices $T$, normalized as in this work, the estimate in the last display has the following version
\begin{gather*}
\pr\left\{\left\|Mx\right\|>A\right\}
\le Ce^{-cA^2W}
\end{gather*}
The same is true for GOE random matrix. In fact \emph{the latter one, i.e. the GOE matrix operator norm upper tail estimateis, is what we use the most}. For Gaussian random non-symmetric matrix \emph{we use vector action concentration estimate in the last display applied to the first standard basis vector}. 
Since we use the estimate repeatedly, it is convenient to state the estimate as a separate proposition. We discuss the proof of the estimate for completeness. 
\begin{prop}\label{prop:GOEopnorm} $(0)$ For Gaussian random vector $x=(\xi_1,...,\xi_W)$, normalized as in this work, holds
\begin{gather*}
\pr\left\{\left|\left\|x\right\|-1\right|>AW^{-{1\over 2}}\right\}
\le Ce^{-cA^2}
\end{gather*}
for any $A>0$ with absolute constants $c,C$.

\smallskip
$(i)$ For Gaussian non-symmetric matrix $T$ holds
\begin{gather*}
\pr\left\{\left|\left\|Tx\right\|-1\right|>AW^{-{1\over 2}}\right\}
\le Ce^{-cA^2}
\end{gather*}
for any unit vector $x\in\IR^W$ and any $A>0$. In particular,
\begin{gather*}
\pr\left\{\left|\left\|Tx\right\|-1\right|>a\right\}
\le Ce^{-ca^2W}
\end{gather*}
for any $a>0$.

\smallskip
$(ii)$ For Gaussian random matrix $T$ holds
\begin{gather*}
\pr\left\{\left\|T\right\|>A\right\}
\le Ce^{-cA^2W}
\end{gather*}
for any $A>0$. 

\smallskip
$(iii)$ For GOE matrix $V$ holds
\begin{gather*}
\pr\left\{\left|\left\|Vx\right\|-1\right|>A\right\}
\le Ce^{-cA^2W}
\end{gather*}   
for any unit vector $x\in\IR^W$ and any $A>0$. In particular,
\begin{gather*}
\pr\left\{\left\|Vx\right\|>1+a\right\}
\le Ce^{-ca^2W}
\end{gather*}
for any $a>0$.  

\smallskip
$(iv)$ For GOE matrix $V$ holds
\begin{gather*}
\pr\left\{\left\|V\right\|>A\right\}
\le Ce^{-cA^2W}
\end{gather*}
for any $A>0$.   
\end{prop}
\begin{proof} $(0)$ Write 
\begin{gather*}
\pr\left\{\left\|x\right\|>1+AW^{-{1\over 2}}\right\}=\pr\left\{\sum_{1\le j\le W}\xi_{i}^2>1+2AW^{-{1\over 2}}+A^2W^{-1}\right\}=\\
\pr\left\{\chi^2_W>W+2AW^{{1\over 2}}+A^2\right\},
\end{gather*}
where $\chi^2_n$ stands for chi--squared distributed with $n$ degrees of freedom.
The Laurent-Massart concentration estimate reads
\begin{gather*}
\pr\left\{\chi^2_n-n>2\sqrt{nx}+2x\right\}<e^{-x},\\
\pr\left\{-\chi^2_n>2\sqrt{nx}\right\}<e^{-x}
\end{gather*}
for any $x>0$, see ~\cite{LM00}. Take $n=W$, use the estimate, get
\begin{gather*}
\pr\left\{\chi^2_W>W+2AW^{{1\over 2}}+A^2\right\}\le
\pr\left\{\chi^2_W-W>2\sqrt{W\times {A^2\over 2}}+2\times{A^2\over 2}\right\}
<e^{-{A^2\over 4}}
\end{gather*}
Similarly,
\begin{gather*}
\pr\left\{\left\|x\right\|<1-AW^{-{1\over 2}}\right\}=
\pr\left\{\chi^2_W<W-2AW^{{1\over 2}}+A^2\right\}\le \pr\left\{\chi^2_W<W-AW^{{1\over 2}}\right\}\le\\
\pr\left\{W-\chi^2_W>2\sqrt{W\times {A^2\over 4}}\right\}
<e^{-{A^2\over 4}},
\end{gather*}
provided $A<W^{1\over 2}$. That finishes the statement in $(0)$ for any $A>0$. 

\smallskip
$(i)$ Use orthogonal right multiplication invariance of $d\pr\big(T\big)$, replace vector $x$ by the standard basis in $\IR^W$ first vector $e_1=(1,0,\dots,0)^t$, write 
\begin{gather*}
Tx=Te_1=(\xi_{1,1},\xi_{2,1},\dots,\xi_{W,1})^t
\end{gather*}
Use part $(0)$, get the statement.

The estimate in $(ii)$ follows from $(i)$ via $\epsilon$--net argument, see ~\cite{TAO-13}, p. 129.

\smallskip
$(iii)$ Use the argument in the proof of $(i)$, get the statement. 

The estimate in $(iv)$ follows from $(iii)$ via $\epsilon$--net argument.
\end{proof}

\smallskip
$\blacktriangle$
As we mentioned in the Introduction, we need a version of Wegner estimate for GOE perturbations of arbitrary symmetric matrix $A$. Any version which is not "too weak" would do. The latter means that the estimate eliminates the \emph{inverse matrix magnitudes which are polynomial in size of the matrix}. On the other hand \emph{an optimal Wegner estimate was established} in \cite{APSSS17}. Out of convenience we state this estimate here. Take GOE matrix $V$, normalized as in this work. Take arbitrary real symmetric matrix $A$. Denote $H=V+A$. Take arbitrary unit vector $f$. Theorem 1 in \cite{APSSS17} says that the following estimate holds for any $K>0$ 

\begin{equation}
\label{eq:Weg}
\pr\left\{\bn H^{-1}\bn> KW\right\}\le {C \over K},
\end{equation}
where $C$ is an absolute constant. We use this estimate with $K=W^{D-1}$, i.e.
\begin{equation}
\label{eq:Weg1}
\pr\left\{\bn H^{-1}\bn> W^D\right\}\le CW^{-\left(D-1\right)},
\end{equation}
where $D\gg 1$ is an absolute constant.

\smallskip
$\blacktriangle$
Next proposition discusses straightforward applications of the estimates from Proposition~\ref{prop:GOEopnorm} and \eqref{eq:Weg1} to Schenker variables. 
\begin{prop}\label{prop:mainest} Use  the notations in Proposition~\ref{prop:newvariables} with $\phi_k=\phi$, $\rho_k=\rho$. 

\smallskip
$(0)$
For any $n$ and $A>0$ 
\begin{gather*}
\pr_N\left\{\left\|T_{n-1}\right\|>A\right\},\pr_N\left\{\left\|V_n\right\|>A\right\}\le Ce^{-cA^2W}
\end{gather*}
with absolute constants $c,C$.

\smallskip
$(1)$ For any $U_{n-2}$ holds
\begin{gather*}
\pr_N\left(\left\{\bn U^{-1}_{n-1}\bn> W^D\right\}|U_{n-2}\right)\le CW^{-\left(D-1\right)},
\end{gather*}
where $C$ is an absolute constant.

\smallskip
$(2)$   Take arbitrary orthogonal matrices $Q_1,Q_2\in\mathbb O_W$, which may depend on the variables $U_{n-1},U_{n}$. Change the variable $T_{n-1}=Q_2S_{n-1}Q_1$ with the rest of the variables being intact. Then $S_{n-1}\in \mathcal F\left(U_{n-1},T_{n-1},U_n\right)$,
\be
\nn d\pr_N\Big(S_{n-1}\Big)=\rho_{n-1}\left(S_{n-1}\right)\times dS_{n-1},\qquad\qquad\quad\qquad\qquad\\
\label{eq:tripletdistributnS}d\pr_N\Big(S_{n-1},U_n\big|\Big) =\qquad\qquad\qquad\qquad\qquad\qquad\\
\nn\qquad\qquad \rho_{n-1}\left(S_{n-1}\right)\times\phi_n\left(U_n+Q_1^tS^t_{n-1}Q_2^t\times U_{n-1}^{-1}\times Q_2S_{n-1}Q_1\right)\times dS_{n-1}\times dU_n
\ee

\smallskip
$(3)$ For any $A>0$
\begin{gather*}
\pr_N\left\{\left\|S_{n-1}\right\|>A\right\}\le Ce^{-c\kappa^2W},
\end{gather*}

\smallskip
$(4)$ Denote:
\begin{gather*}
\textsf V_{n}=\textsf V_{n}\left(S_{n-1};U_{n-1},U_{n}\right)=U_n+Q_1^tS^t_{n-1}Q_2^t\times U_{n-1}^{-1}\times Q_2S_{n-1}Q_1,\\
\cB_{S,n-1,A}=\left\{\left\|S_{n-1}\right\|>A\right\},\quad \cB_{V,n,A}=\left\{\left\|V_n\right\|>A\right\},\\
\cB_{S,n-1,n,A}:=\cB_{S,n-1,n,A,\big|U_{n-1}}=\cB_{S,n-1,A}\cup
\left\{\left(S_{n-1},U_n\right):\textsf V_{n}\in\cB_{V,n,A}\right\}
\end{gather*}
Then
\begin{gather*}
\pr_N\left[\cB_{S,n-1,n,A,\big|U_{n-1}}\big|U_{n-1}\right]\le 2Ce^{-cA^2W}
\end{gather*}

\smallskip
$(5)$ There exists a set $\cB_{U,n-1,n,A}\in \mathcal F\left(U_{n-1},U_{n}\right)$ such that
\begin{gather*}
\pr_N\left[\cB_{U,n-1,n,A}\right]\le C_1e^{-c_1A^2W},\\
\pr_N\left[\cB_{S,n-1,n,A}\big|U_{n-1},U_{n}\right]\le C_1e^{-c_1A^2W},\quad \text{for any $\left(U_{n-1},U_{n}\right)\notin \cB_{U,n-1,n,A}$},
\end{gather*}
with absolute constants $c_1,C_1$.

\smallskip
$(6)$ Use the notations in $(5)$. Denote:
\begin{gather*}
\widetilde{\textsf V}_{n}:=\widetilde{\textsf V}_{n}\left(S_{n-1};U_{n-1},U_{n}\right)=Q_1U_nQ_1^t+S^t_{n-1}Q_2^t\times U_{n-1}^{-1}\times Q_2S_{n-1},\\
\widetilde{\cB}_{S,n-1,n,A}=\cB_{S,n-1,A}\cup
\left\{\left\|\widetilde{\textsf V}_{n}\right\|>A\right\},
\end{gather*}
Then
\begin{gather*}
\pr_N\left[\widetilde{\cB}_{S,n-1,n,A}\big|U_{n-1},U_{n}\right]\le C_1e^{-c_1A^2W},\quad \text{for any $\left(U_{n-1},U_{n}\right)\notin \cB_{U,n-1,n,A}$}
\end{gather*}

\end{prop}
\begin{proof} $(0)$ The estimates follow from Proposition~\ref{prop:GOEopnorm} parts $(ii)$ and $(iv)$ respectively. 

\smallskip
$(1)$ By the definition $U_{n-1}=V_{n-1}-T^t_{n-2}\times U_{n-2}^{-1}\times T_{n-2}$. Proposition~\ref{prop:condecoupling1} says that the randdom variables $V_{n-1}$ and $T^t_{n-2}\times U_{n-2}^{-1}\times T_{n-2}$ are independent. Take $H=V+A$, where $V$ is GOE distributed, and $A=-T^t_{n-2}\times U_{n-2}^{-1}\times T_{n-2}$. Apply \eqref{eq:Weg1}, get the statement.

\smallskip
$(2)$ Lemma~\ref{lem:elementary3} says that the Jacobian of the variable change $T_{n-1}=S_{n-1}Q$ is $=1$, including the case when $Q$ depends on the rest of variables. From the definitions in $(3)$ follows $S_{n-1}\in \mathcal F\left(U_{n-1},T_{n-1},U_n\right)$. Combine  the change of variables distribution density rule with Jacobian $=1$, with the equations in \eqref{eq:tripletdistributn1}, and with invariance of $\rho_{n-1}$ under the change of variable in question, get the statements in $(3)$. 

\smallskip
$(3)$ From the definition $T_{n-1}=Q_2S_{n-1}Q_1$. Therefore $\left\|T_{n-1}\right\|=\left\|S_{n-1}\right\|$ and the statement follows from  part $(1)$.

\smallskip
$(4)$
Use \eqref{eq:tripletdistributnS}, write
\begin{gather*}
\pr_N\left[\cB\big|U_{n-1}\right]=\int_{\cB}d\pr_N\Big(S_{n-1},U_n\big|U_{n-1}\Big)=\\
\int_{\cB}\rho_{n-1}\left(S_{n-1}\right)\times\phi_n\left(U_n+Q_1^tS^t_{n-1}Q_2^t\times U_{n-1}^{-1}\times Q_2S_{n-1}Q_1\right)\times dT_{n-1}\times dU_n
\end{gather*}
In this dispaly $\cB:=\cB_{U,S,n,A,Q\big|U_{n-1}}$ to save space. 
In the display second line integral change back the variable $S_{n-1}=Q_2^tT_{n-1}Q_1^t$. Combine  the integration change of variables rule with Jacobian $=1$ and with invariance of $\rho_{n-1}$, get 
\begin{gather*}
\pr_N\left[\cB_{U,S,n,A,Q\big|U_{n-1}}\big|U_{n-1}\right]=\\
\int_{\cB_{U,T,n,A\big|U_{n-1}}}\rho_{n-1}\left(T_{n-1}\right)\times\phi_n\left(U_n+T^t_{n-1}\times U_{n-1}^{-1}\times T_{n-1}\right)\times dT_{n-1}\times dU_n
\end{gather*}
Denote:
\begin{gather*}
\cB_{U,T,n,A,Q\big|U_{n-1}}=\left\{\Big(T_{n-1},U_n\Big):U_n+T^t_{n-1}\times U_{n-1}^{-1}\times T_{n-1}\in\cB_{V,n,\kappa}\right\}
\end{gather*}
In the last integral change the variable $U_n+T^t_{n-1}\times U_{n-1}^{-1}\times T_{n-1}=V_n$ and keep the variable $T_{n-1}$ intact. The Jacobian of this 
change of the variables is $=1$. From the definitions the image of $\cB_{U,T,n,\kappa,Q\big|U_{n-1}}$ under this change of variables is $\cB_{V,n,\kappa}$. Write
\begin{gather*}
\pr_N\left[\cB_{U,S,n,A,Q\big|U_{n-1}}\big|U_{n-1}\right]=
\int_{\cB_{V,n,A}}\rho_{n-1}\left(T_{n-1}\right)\times\phi_n\left(V_n\right)\times dT_{n-1}\times dV_n\\\le
\pr_N\left[ \cB_{S,n-1,A}\right]+\pr_N\left[ \cB_{V,n,A}\right]\le 2Ce^{-cA^2W}
\end{gather*}
Here the last inequality in the display comes from part $(1)$.

\smallskip
$(5)$ 
From the definitions
\begin{gather*}
\left[\cB_{S,n-1,n,A}\right]_{\big|U_{n-1},U_n}=\cB_{S,n-1,n,A,\big|U_{n-1}},
\end{gather*}
where $\cB_{\big|\cdot}$ stands for the cross-section of the set $\cB$ at a given point $\cdot$ from the respective variable domain. Combine, use the standard disintegration argument with help of Chebyshev inequality, get the statement.

\smallskip
$(6)$ Use 
\begin{gather*}
\widetilde{\textsf V}_{n}=Q_1\textsf V_{n}Q^t_1,\quad \big\|\widetilde{\textsf V}_{n}\big\|=\left\|\textsf V_{n}\right\|,
\end{gather*}
combine with $(5)$, get the statements in $(6)$.

\end{proof}

\begin{proof}[Proof of Proposition ~\ref{prop:condecoupling}]

\smallskip
$(0)$  Use the notations in part $(3)$ of Proposition ~\ref{prop:condecoupling1}. The spherical variables change Jacobian is $=r^{W-1}$. Combine the density change of variables rule with distribution $d\pr\left(\left(\xi_{1,1},\tilde X^t_1\right)^t\big|\right)$ equations from  part $(3)$ of Proposition ~\ref{prop:condecoupling1}, compute the joint distribution $d\pr\left(r,\textsf x\right)$. Use conditional distribution rule, get the statements in $(0)$.

\smallskip
$(1)$ Proposition ~\ref{prop:condecoupling1}, $(1)$ says that for any $U_{n-2}$ holds
\begin{gather*}
\pr_N\left(\left\{\bn U^{-1}_{n-1}\bn> W^D\right\}|U_{n-2}\right)\le CW^{-\left(D-1\right)},
\end{gather*}
It is convenient to denote $\widehat{\cB}_{U,n-2,n-1,D,A}\in \mathcal F\left(U_{n-1},U_{n}\right)$ the set of pairs $\left(U_{n-2},U_{n-1}\right)$ in the display, though $U_{n-2}$ is actually arbitrary . Use the definitions, get 
\begin{gather*}
\left\|C\right\|=\left\|Q_2^t\times U_{n-1}^{-1}\times Q_2\right\|=\left\|U_{n-1}^{-1}\right\|\le CW^{-\left(D-1\right)},
\end{gather*}
for any $\left(U_{n-2},U_{n-1}\right)\notin \widehat{\cB}_{U,n-2,n-1,D,A}$. Use Proposition~\ref{prop:mainest}, $(6)$
\begin{gather*}
\pr_N\left\{\left\|S_{n-1}\right\|>A\right\}\le C_1e^{-c_1A^2W}
\end{gather*}
with absolute constants $c_1,C_1$. Use the definitions of $\textsf a^2\left(\textsf x\right),\textsf b\left(\textsf x\right),\textsf c\left(\textsf x\right)$, combine, get the first statement in $(1)$. For any symmetric matrix $M\neq 0$ holds
\begin{gather*}
\mes_{n-1}\left\{\textsf x\in \mathbb S^{n-1}: \langle Mx,x\rangle=0\right\}=0,
\end{gather*}
where $\mes_{n-1}$ stands for the normalized Lebesgue measure on the unit sphere $\mathbb S^{n-1}$. It follows from the definitions that $C$ is actually invertible with probability $1$. That implies the second statement in $(1)$.

\smallskip
$(2)$ Use Proposition~\ref{prop:mainest}, $(6)$: for any $A>0$ there exists a set $\cB_{U,n-1,n,A}\in \mathcal F\left(U_{n-1},U_{n}\right)$, such that
\begin{gather*}
\pr_N\left[\cB_{U,n-1,n,A}\right]\le C_1e^{-c_1A^2W},\\
\pr_N\left[\widetilde{\cB}_{S,n-1,n,A}\big|U_{n-1},U_{n}\right]\le C_1e^{-c_1A^2W},\quad \text{for any $\left(U_{n-1},U_{n}\right)\notin \cB_{U,n-1,n,A}$},
\end{gather*}
with absolute constants $c_1,C_1$, where 
\be\nn
\widetilde{\cB}_{S,n-1,n,A}:=\left\{\left\|S_{n-1}\right\|>A\right\}\cup
\left\{\left\|\widetilde{\textsf V}_{n}\right\|>A\right\},\qquad\qquad\qquad\qquad\\
\label{eq:tilde BV}\\
\nn\widetilde{\textsf V}_{n}:=\widetilde{\textsf V}_{n}\left(S;U_{n-1},U_{n}\right)=Q_1U_nQ_1^t+S^tQ_2^t\times U_{n-1}^{-1}\times Q_2S
\ee
Redenote  
\be
\label{eq:tilde BVcross}
\cB_{\cX,n-1,n,A}=\widetilde{\cB}_{\big|U_{n-1},U_{n}}
\ee
Here, as before, $\cB_{\big|\cdot}$ stands for the cross-section of the set $\cB$ at a given point $\cdot$ from the respective variable domain. 
The probability estimates in part $(6)$ hold for these sets. Verify the equation in part $(5)$ of the proposition with these sets in place and with
$$\left(U_{n-1},U_{n}\right)\notin \cB_{U,n-1,n,A},\quad X\notin \cB_{\cX,n-1,n,A}$$

From the definition the vector $X:=\left(\xi_{1,1},\tilde X^t_1\right)^t$ is the first column of $S$. Therefore,
\be\label{eq:rvariable}
r=\left\|\left(\xi_{1,1},\tilde X^t_1\right)^t\right\|\le \left\|S\right\|\le A,
\ee 
as claimed. Similarly,
\be\label{eq:tilde V}
\left\|\widetilde{\textsf V}_{n}e_1\right\|\le \left\|\widetilde{\textsf V}_{n}\right\|\le A
\ee

From the definitions in part $(2)$ 
\be\label{eq:mainsplitaC}
\lambda_1\xi^2_{1,1}+\langle \check C\tilde X_1,\tilde X_1\rangle=\langle CX,X\rangle
\ee
Next, use equation \eqref{eq:mainsplit} from part $(2)$ proof, combine with the definition of $\widetilde{\textsf V}_{n}$, see above, write 
\be\nn
\widetilde{\textsf V}_{n}=Q_1U_nQ_1^t+S_{n-1}^t\times U_{n-1}^{-1}\times S_{n-1}=B+S_{n-1}^t\times U_{n-1}^{-1}\times S_{n-1}=\\
\nn\left(b_{i,j}\right)_{1\leq i,j\leq W}+\begin{bmatrix}
\lambda_1\xi^2_{1,1}+\langle \check C\tilde X_1,\tilde X_1\rangle&\lambda_1\xi_{1,1}\tilde Y_1+\left(\check{T}^t\check C\tilde X_1\right)^t\\
\lambda_1\xi_{1,1}\tilde Y^t_1+\check{T}^t\check C\tilde X_1&\lambda_1\tilde Y^t_1\times\tilde Y_1+ \check{T}^t\check C\check{T}\end{bmatrix}
\ee
In particular,
\be\label{eq:mainsplita}
\widetilde{\textsf V}_{n}e_1=
\left(b_{i,1}\right)_{1\leq i,j\leq W}+\begin{bmatrix}
\langle CX,X\rangle\\
\lambda_1\xi_{1,1}\tilde Y^t_1+\check{T}^tC\tilde X_1\end{bmatrix}
\ee 
Use the equations in part $(0)$, combine with \eqref{eq:rvariable},\eqref{eq:mainsplitaC}, \eqref{eq:mainsplita}, get
\begin{gather*}
b_{1,1}=\langle\widetilde{\textsf V}_{n}e_1,e_1\rangle-\langle CX,X\rangle,\\
\gamma\left(\textsf x\right) r^2=b_{1,1}\langle C\textsf x,\textsf x\rangle r^2=b_{1,1}\langle CX,X\rangle=\langle\widetilde{\textsf V}_{n}e_1,e_1\rangle\langle CX,X\rangle-\langle CX,X\rangle^2
=\\
\sqrt 2\langle\widetilde{\textsf V}_{n}e_1,e_1\rangle\textsf ar^2
-2\textsf a^2r^4,\\
\bm\gamma\left(\textsf x\right)r^2+2\textsf a^2\left(\textsf x\right)r^4\bm=\sqrt 2\bm\langle\widetilde{\textsf V}_{n}e_1,e_1\rangle\bm\textsf a\left(\textsf x\right)r^2\le \sqrt 2 A \textsf a\left(\textsf x\right)r^2,\\
\left(\langle\widetilde{\textsf V}_{n}e_1,e_1\rangle\right)_{2\leq i\leq W}=\tilde B_1+\lambda_1\xi_{1,1}\tilde Y^t_1+\check{T}^tC\tilde X_1,\\
r^2\beta^2\left(\textsf x\right)=r^2\left\|\lambda_{1}\textsf x_{1}\tilde Y^t_1+\check{S}^t\check C\tilde{\textsf x}\right\|^2=\left\|\lambda_{1}\xi_{1,1}\tilde Y^t_1+\check{S}^t\check C\tilde{X}_1\right\|^2,\\
\textsf c\left(\textsf x\right)r=
2r\langle\lambda_{1}\textsf x_{1}\tilde Y^t_1+\check{S}^t\check C\tilde{\textsf x},\tilde B_1\rangle=2\langle\lambda_{1}\xi_{1,1}\tilde Y^t_1+\check{S}^t\check C\tilde{X}_1,\tilde B_1\rangle=\\
2\langle\left(\langle\widetilde{\textsf V}_{n}e_1,e_i\rangle\right)_{2\leq i\leq W},\lambda_{1}\xi_{1,1}\tilde Y^t_1+\check{S}^t\check C\tilde{X}_1\rangle-2\left\|\lambda_{1}\xi_{1,1}\tilde Y^t_1+\check{S}^t\check C\tilde{X}_1\right\|^2=\\
2\langle\left(\langle\widetilde{\textsf V}_{n}e_1,e_i\rangle\right)_{2\leq i\leq W},\lambda_{1}\xi_{1,1}\tilde Y^t_1+\check{S}^t\check C\tilde{X}_1\rangle-2\beta^2\left(\textsf x\right)r^2,\\
\left|\textsf c\left(\textsf x\right)r+2\beta^2\left(\textsf x\right)r^2\right|=2\left|\langle\left(\langle\widetilde{\textsf V}_{n}e_1,e_i\rangle\right)_{2\leq i\leq W},\lambda_{1}\xi_{1,1}\tilde Y^t_1+\check{S}^t\check C\tilde{X}_1\rangle\right|\le \\
2\left\|\widetilde{\textsf V}_{n}\right\|\times \left\|\lambda_{1}\xi_{1,1}\tilde Y^t_1+\check{S}^t\check C\tilde{X}_1\right\|\le 2A\beta\left(\textsf x\right)r
\end{gather*}
The verification of the equations
\begin{gather*}
\textsf ar^2< A+\left\|B_{1}\right\|,\\
|\gamma|r^2\le \left\|B_{1}\right\|\left(A+\left\|B_{1}\right\|\right),\\
\beta^2r^2\le \left( A+\left\|B_{1}\right\|\right)^2,\\
|b|r^2\le A^2+\left( A+\left\|B_{1}\right\|\right)^2+\left\|B_{1}\right\|\left(A+\left\|B_{1}\right\|\right),\\
|\textsf c|r\le 2 \left\|B_{1}\right\|\left(A+\left\|B_{1}\right\|\right)
\end{gather*}
is similar and actually a bit shorter. We skip it. That finishes the proof of part $(2)$.

\smallskip
$(3)$ Use part $(2)$, combine with the standard disintegration argument and with Chebyshev inequality, get the statements in $(3)$.

\smallskip
$(4)$ Use the definitions, get
\begin{gather*}
\textsf a^2={1\over 2}\langle C\textsf x,\textsf x\rangle^2\le {1\over 2}\left\|C\textsf x\right\|^2=
{1\over 2}\left\|X\right\|^{-2}\left\|CX\right\|^2={1\over 2}\zeta^2
\end{gather*}
Next, use 
\begin{gather*}
S^tCX=
\begin{bmatrix}
\lambda_1\xi^2_{1,1}+\langle \check C\tilde X_1,\tilde X_1\rangle\\
\lambda_1\xi_{1,1}\tilde Y^t_1+\check{S}^t\check C\tilde X_1\end{bmatrix},
\end{gather*}
see Proposition~\ref{prop:condecoupling1}, part $(3)$. Use the definitions, combine with the equation \eqref{eq:rvariable}, write
\begin{gather*}
\beta^2r^2=\left\|\lambda_{1}\xi_{1,1}\tilde Y^t_1+\check{S}^t\check C \tilde X_1\right\|^2\le \left\|S^tCX\right\|^2\le A^2\left\|CX\right\|^2=A^2\left\|\textsf Z\right\|^2=A^2\zeta^2r^2,\\
\beta^2\le A^2\zeta^2
\end{gather*}
as claimed.
Finally, take $\textsf x\notin \widehat{\cB}_{\texttt x,n-1,n,A}$, $r\notin \cB_{\texttt r,n-1,n,A}$ as in $(3)$, write
\begin{gather*}
\bm\gamma r^2\bm\le \bm\gamma r^2+2\textsf a^2r^4\bm+2\textsf a^2r^4\le 
\sqrt 2A\textsf ar^2+2\textsf a^2r^4\le A\zeta r^2+\zeta^2r^4,\\
\left|\textsf c\right|r\le 2\beta^2r^2+\left|\textsf cr+2\beta^2r^2\right|\le 2A^2\zeta^2r^2+2A\zeta r
\end{gather*}
Use that $r\le A$ if $r\notin \cB_{\texttt r,n-1,n,A}$, conclude
\begin{gather*}
\bm\gamma\bm \le A\zeta+\zeta^2r^2 \le A\zeta+A^2\zeta^2,\\
\left|\textsf b\right|\le 1+\beta^2+\bm\gamma\bm\le 1+A^2\zeta^2+A\zeta+A^2\zeta^2= 1+A\zeta+2A^2\zeta^2,\\
\left|\textsf c\right|\le 2A^3\zeta^2+2A\zeta,
\end{gather*}
as claimed. Since $r$ does not enter the estimates in the last display the estimates \underline{hold without condition $r\notin \cB_{\texttt r,\textsf x,n-1,n,A}$}.

\smallskip
$(5)$ 
Rewrite the recurrence equation \eqref{eq:recurrence} from Proposition~\ref{prop:condecoupling1} part $(4)$ as follows 
\be
\nn \left\|\textsf B_{n-2}\left(\textsf y\right)\right\|=\left\|X_{1,n-1}\right\|\left\|\textsf Z_{n-1}\right\|^{-1}=\zeta_{n-1}^{-1}\left(\textsf x\right)
\ee
Take 
\begin{gather*}
\textsf y\notin \widehat{\cB}_{\texttt x,n-2,n-1,A},\quad \rho\notin \cB_{\texttt r,\textsf x,n-2,n-1,A}
\end{gather*}
Use the part $(2)$ last display estimates with $(n-2)$ in the role of $(n-1)$. Use $\rho$ for the notation for the sperical radius and $\textsf y$ for the spherical variable. Take 
\begin{gather*}
\textsf y\notin \widehat{\cB}_{\texttt x,n-2,n-1,A},\quad \rho\notin \cB_{\texttt r,\textsf x,n-2,n-1,A},
\end{gather*}
write
\begin{gather*}
\textsf a_{n-2}\left(\textsf y\right)\rho^2< A+\left\|\textsf B_{n-2}\left(\textsf y\right)\right\|,\\
\beta_{n-2}^2\left(\textsf y\right)\rho\le \left( A+\left\|\textsf B_{n-2}\left(\textsf y\right)\right\|\right)^2,\\
|\textsf c_{n-2}\left(\textsf y\right)|\rho\le 2 \left\|\textsf B_{n-2}\left(\textsf y\right)\right\|\left(A+\left\|\textsf B_{n-2}\left(\textsf y\right)\right\|\right)
\end{gather*}
If 
$$\zeta_{n-1}\left(\textsf x\right)>\textsf C,$$
then
\begin{gather*}
\left\|\textsf B_{n-2}\left(\textsf y\right)\right\|<\textsf C^{-1},\\
\textsf a_{n-2}\left(\textsf y\right)\rho^2< A+\textsf C^{-1},\\
|\gamma_{n-2}\left(\textsf y\right)|\rho^2\le \textsf C^{-1}\left(A+\textsf C^{-1}\right),\\
\beta^2_{n-2}\left(\textsf y\right)\rho^2\le \left( A+\textsf C^{-1}\right)^2,\\
\textsf b_{n-2}\left(\textsf y\right)\rho^2\le A^2+2\left( A+\textsf C^{-1}\right)^2,\\
|\textsf c_{n-2}\left(\textsf y\right)|\rho\le 2 \textsf C^{-1}\left(A+\textsf C^{-1}\right)
\end{gather*}
\emph{regardless how large is $\textsf C$}. 
Now we got to verify that dichotomy for $\textsf b:=\textsf b_{n-2}\left(\textsf y\right)$. Apply the current proposition part $(0)$ and Proposition~\ref{prop:newvariables}  part $(2)$ with $(n-2)$ in the role $(n-1)$. It is  convenient though to consider $(n-1)$ instead of $(n-2)$ That way we just use the same equations without any comments. Use 
\begin{gather*}
\textsf a^2\left(\textsf x\right)={1\over 2}\langle C\textsf x,\textsf x\rangle^2,\\
\textsf b=\textsf b\left(\textsf x\right)=1+\beta^2+\gamma ,\\
\gamma\left(\textsf x\right)=b_{1,1}\langle C\textsf x,\textsf x\rangle,\\
|\gamma\left(\textsf x\right)|=|b_{1,1}|\left|\langle C\textsf x,\textsf x\rangle\right|<\textsf C^{-1}\left|\langle C\textsf x,\textsf x\rangle\right|=\textsf C^{-1}\sqrt 2|\textsf a\left(\textsf x\right)|
\end{gather*}
see part $(0)$ in the current proposition. If $\textsf b\left(\textsf x\right)>0$, then we are done. Assume $\textsf b\left(\textsf x\right)\le 0$.
Then the last dispay implies
\begin{gather*}
1+\beta^2<|\gamma|<\textsf C^{-1}\sqrt 2|\textsf a\left(\textsf x\right)|,\\
|\textsf b\left(\textsf x\right)|\le 1+\beta^2+|\gamma\left(\textsf x\right)|<2|\gamma\left(\textsf x\right)|<2\textsf C^{-1}\sqrt 2|\textsf a\left(\textsf x\right)|,
\end{gather*}
as claimed.

\end{proof}

\section{Concentration Estimates and Log-Variance of Spherical Radius Against Special Super-Exponential Densitities.}\label{sec:superexp}
The proof of Theorem B is split into few proposition. Each of these propositions targets a specific technological problem which enter our method.  The first target consists of concentration estimates and log-variance of spherical radius against super-exponential densitities from Proposition ~\ref{prop:condecoupling}, part 4,
\be\label{eq:radiusdistrib}
d\pr(r)=\mu r^{W-1}e^{-W\left(\textsf a^2 r^4+\textsf br^2+\textsf c r\right)}d r:=\mu\phi(r)dr,\quad r\ge 0
\ee
where $\textsf a$, $\textsf b$, $\textsf c$ are constant coefficients, $\textsf a^2>0$, $\mu=\mu\left(\textsf a,\textsf b,\textsf c\right)$ is the normalizing factor.  Actually, in Proposition ~\ref{prop:condecoupling} the coefficients $\textsf a$, $\textsf b$ , $\textsf c$ \emph{depend on the spherical variable $\textsf x\in\mathbb S^{W-1}$}. The coefficients depend also on some other variables too. On the other hand Proposition ~\ref{prop:condecoupling} says that off a set of very small probability $($ the actual estimate of the set in question will play an important role in Section~\ref{sec:Theorem B}, but not in the curent section$)$, the coefficient obey some very important estimates. In the current section we list some of these estimates as the main conditions under which the result of the section hold.

Here is the list of conditions in question.

$(I)$ Crude upper bound estimates:
\begin{gather*}
\max\left(\textsf a^2,\left|\textsf b\right|,\left|\textsf c\right|\right)\le 2A^2W^{2D},
\end{gather*}
with absolute constants $A,D$

$(II)$ The distribution "essential" support size estimate 
\be\label{eq:rsupport}
\pr\left\{r\ge A\right\}\le C_0e^{-c_0A^2W},
\ee
with absolute constants $c_0,C_0$. Here $\pr$ stands for the probability distribution in\eqref{eq:radiusdistrib}.

$(III)$ The dichotomy conditions. Either
\be
\label{eq:dichotomy1}\max\left(\textsf a^2,\left|\textsf b\right|,\left|\textsf c\right|\right)\le \textsf A
\ee
with absolute constant $\textsf A$, or there exists a set $\cB_{\texttt r}$
\begin{gather*}
\pr\cB_{\texttt r}\le C_0e^{-c_0A^2W},
\end{gather*}
such that for any $r\notin \cB_{\texttt r}$ holds
\be
\max\left(\textsf a^2r^4,\textsf br^2,\left|\textsf c\right|r\right)\le \textsf A\label{eq:dichotomy2}
\ee
Furthermore, in this case, either
\be\label{eq:bcrucialaa}
\textsf b>0, 
\ee
or 
\be\label{eq:bcrucialdaa}
|\textsf b|<\textsf C^{-1}|\textsf a|
\ee

 The main result of this section is the following proposition
\begin{prop}\label{prop:superexplogvar} Assume that conditions $(I)-(III)$ hold. Then
\be
\Var \log r \ge cW^{-1}\label{eq:superexplogvar}
\ee
with absolute constant $c>0$. Furthemore, there exists a set $\cB_{\phi}\subset (0,+\infty)$ such that
\be
\nn \pr \cB_{\phi}\le Ce^{-cW}\qquad\qquad\qquad\qquad\\
\label{eq:radiuslowerb}\\
\nn r\ge c_1W^{-\textsf d},\quad \text{for any $r\in (0,+\infty)\setminus \cB_{\phi}$},
\ee
where $c,c_1,C,\textsf d$ are absolute constants.
\end{prop}
The proof of Proposition~\ref{prop:superexplogvar} is a simple corollary of the Proposition~\ref{prop:elementary2g} below. The latter proposition 
establishes a detailed analysis of the concentration estimates addressing several different narrow intervals of size $\sim W^{-{1\over 2}}$ of the random variable $r$ values. The final result of the analysis itself is a combination of a very long list of elementary calculus estimates, which address integrals of the density $\phi(r)$ over specific intervals where $\phi$ is monotonic.
\begin{prop}\label{prop:elementary2g} Assume that conditions $(I)-(III)$ hold. Denote
\be
\nn f\left(r\right)=W^{-1}\log \phi\left(r\right)=\qquad\qquad\qquad\qquad\\
\label{eq:radiusdistriba}\\
\nn W^{-1}\log \mu+\left(1-W^{-1}\right) \log r-\textsf a^2r^4-\textsf br^2-\textsf cr,\quad r\ge 0
\ee

\smallskip
$(i)$ There exist $0<r_1<2A$ and $\delta_1>0$, such that 
$f'\left(r\right)\ge 0$ for $0\le r\le r_1$, and $f'\left(r\right)<0$ for $r_1<r\le r_1+\delta_1$. In particular, $f'\left(r_1\right)=0$, $f''\left(r_1\right)<0$. The estimates 
\begin{gather*}
r_1\ge c_0\min\left(\textsf a^{-{1\over 2}},|\textsf b|^{-{1\over 2}},|\textsf c|^{-1}\right)\ge cW^{-\textsf d}
\end{gather*}
hold with absolute constants $c,\textsf d$.

\smallskip
$(ii)$ 
Assume that $\textsf b>0$. Then $r_1$ is the only critical point of $f$ on the interval $[0,2A]$.

\smallskip
$(iii)$ There might be at most one point $0<r_1<r_3$ with $f'\left(r_3\right)=0$, $f''\left(r_3\right)\le 0$. If $r_3$ exists, then $(iiia)$ $f''\left(r_3\right)<0$, $(iiib)$ $f(r)$ has a local minimum at $r_1< r_2<r_3$ and no other critical points on $(0,+\infty)$, $(iiic)$ $f''(r)$ has exactly two zeros, both sitting in $[r_1,r_3]$ and no other positive zeros.

\smallskip
$(iv)$ Take $r_j$, $j\in\{1,3\}$ and $\eta>-1$. The following equation holds
\be
\nn f\left(r_j\left(1+\eta\right)\right)- f\left(r_j\right)=\qquad\qquad\qquad\qquad\\
\label{eq:logdensitydecay}\\
\nn\left(1-\nu\right) \log \left(1+\eta\right)-\left(1-\nu\right)\eta-\left(6\textsf a^2r_j^4+\textsf br_1^2\right)\eta^2-4\textsf a^2r_j^4\eta^3-\textsf a^2r_j^4\eta^4
\ee
In particular, for $|\eta|\le {1\over 4}$
\be
\nn \exp W\left[-{2\over 3}\left(1-\nu\right)\eta^2-6\textsf a^2r_j^4\left(1+{3\eta\over 2}\eta+{\eta^2\over 6}\right)\eta^2-\textsf br_j^2\eta^2\right]\times \\
\label{eq:logdensitydecay2}\phi\left(r_j\right)\le\phi\left(r_j\left(1+\eta\right)\right)\le \phi\left(r_j\right)\times\qquad\qquad\\
\nn\exp W\left[-{2\over 5}\left(1-\nu\right)\eta^2-6\textsf a^2r_j^4\left(1+{3\eta\over 2}+{\eta^2\over 6}\right)\eta^2-\textsf br_j^2\eta^2\right]\quad
\ee

\smallskip
$(v)$ If $\textsf b\ge 0$, then the unction $f(r)$ decreases monotonically on each of the intervals $({3\over 4}r_1,r_1)$ and $(r_1, {5\over 4}r_1)$.  If $\textsf b<0$, then for $|\eta|\le {1\over 4}$
\be\label{eq:logdensitydecay3}\qquad
\exp \left[-W\textsf A_1\eta^2\right]\times 
\phi\left(r_j\right)\le\phi\left(r_j\left(1+\eta\right)\right)\le \phi\left(r_j\right)\times\exp\left[W\textsf A_1\eta^2\right]\quad
\ee
with absolute constant $\textsf A_1>0$.

\smallskip
$(vi)$ Assume $\textsf b\ge 0$. Then
for $|\eta|\le {1\over 8}$ holds
\begin{gather*}
{\phi\left(r_1\left(1+2\eta\right)\right)\over \phi\left(r_1\left(1+\eta\right)\right)}<
\exp W\left[-{3\over 5}\left(1-\nu\right)\eta^2-\textsf a^2r_1^4\eta^2-\textsf br_1^2\eta^2\right],
\end{gather*}

\smallskip
$(vii)$  Assume $\textsf b\ge 0$. There exists unique $0<\phi_{\textsf m}<\phi(r_1)$, such that
\begin{gather*}
\int_{\phi(r)<\phi_{\textsf m}}\phi\left(r\right)dr=\int_{\phi(r)>\phi_{\textsf m}}\phi\left(r\right)dr
\end{gather*}
The set $I_\textsf m=\{\phi(r)>\phi_{\textsf m}\}$ is an interval, $I_\textsf m:=(r_{1,<},r_{1,>})$.

\smallskip
$(viii)$  Assume $\textsf b\ge 0$. Then
\be\label{eq:arbound1}
 \max\left[\textsf a^2r_1^4,\textsf b r_1^2\right]<8A^2
\ee

\smallskip
$(ix)$ Use the notations in $(vii)$. Assume that
\begin{gather*}
\int_{r_1<r<r_{1,>}}\phi\left(r\right)dr\ge \int_{r_{1,<}<r<r_1}\phi\left(r\right)dr
\end{gather*}
Given $0<\tau<1$ there exists unique $r_1<r_{>|\tau}<r_{1,>}$, such that
\begin{gather*}
\int_{r_1<r<r_{>|\tau}}\phi\left(r\right)dr=\tau\int_{r_1<r<r_{1,>}}\phi\left(r\right)dr
\end{gather*}
Furthermore, 
\begin{gather*}
r_{>|{j+1\over 5}}>r_{>|{j\over 5}}\left(1+{c_j\over\sqrt W}\right),\quad j=0,1,2,3 
\end{gather*}
where $c_j>0$ are absolute constants.

A completely similar statement holds in case 
\begin{gather*}
\int_{r_{1,<}<r<r_1}\phi\left(r\right)dr\ge \int_{r_1<r<r_{1,>}}\phi\left(r\right)dr
\end{gather*}
with the interval $(r_{1,<},r_1)$ in the role of $(r_1,r_{1,>})$ and $-\eta$ in the role of $\eta$. We denote the respective points as $r_{<|{j+1\over 5}}<r_{<|{j\over 5}}$, $j=0,1,2,3$ and the main part of the statement is that 
\begin{gather*}
r_{>|{j+1\over 5}}<r_{>|{j\over 5}}\left(1-{c_j\over\sqrt W}\right),\quad j=0,1,2,3 
\end{gather*}
In particular, the statement implies that $r_1-r_{1,<}\ge r_1\left(1-{c\over\sqrt W}\right)$ with an absolute constant $c>0$.

\smallskip
$(x)$ Assume $\textsf b<0$. Assume also there is one local maximum $r_1<2A$. Assume that
\begin{gather*}
\int_{0<r<r_1}\phi\left(r\right)dr\le \int_{r_1<r<+\infty}\phi\left(r\right)dr
\end{gather*}
Given $0<\tau<1$ there exists unique $r_1<r_{>|\tau}$, such that
\begin{gather*}
\int_{r_1<r<r_{>|\tau}}\phi\left(r\right)dr=\tau\int_{r_1<r<+\infty}\phi\left(r\right)dr
\end{gather*}
Furthermore, 
\begin{gather*}
r_{>|{j+1\over 5}}>r_{>|{j\over 5}}\left(1+{c_j\over\sqrt W}\right),\quad j=0,1,2,3 
\end{gather*}
where $c_j>0$ are absolute constants.

A completely similar statement holds in case 
\begin{gather*}
\int_{0<r<r_1}\phi\left(r\right)dr\le \int_{r_1<r<+\infty}\phi\left(r\right)dr
\end{gather*}

\smallskip
$(xi)$ Assume $\textsf b<0$. Assume there are two local maxima $r_1<r_3$ as in part $(3)$.  Assume 
\begin{gather*}
\int_{0<r<r_2}\phi\left(r\right)dr\ge \int_{r_2<r<\infty}\phi\left(r\right)dr,
\end{gather*}
where $r_1<r_2<r_3$ is the only local minimum of $f(r)$, see part $(iii)$. Then all the statements in part $(x)$ hold. The same is true in the complementary case
\begin{gather*}
\int_{0<r<r_2}\phi\left(r\right)dr<\int_{r_2<r<\infty}\phi\left(r\right)dr,
\end{gather*}
with $r_3$ in the role of $r_1$.

\end{prop}
\begin{proof} 
$(i)$ Write
\begin{gather*}
f\left(0\right)=-\infty=f\left(+\infty\right),\\
f'={1-W^{-1}\over r}-\left(4\textsf a^2 r^3+2\textsf b r+\textsf c\right)=-{4\textsf a^2 r^4+2\textsf b r^2+\textsf c r-\left(1-\nu\right)\over r},\\
\nu=W^{-1}\ll 1,\quad
f'\left(0\right)=\infty=-f'\left(+\infty\right)
\end{gather*}
Assume $f'\left(r\right)\ge 0$ for $0\le r\le 2A$. Then $f'\left(r\right)>0$ for $0\le r\le 2A$, and therefore $\phi\left(r\right)$ increases for $0\le r\le 2A$. In particular,
\begin{gather*}
\pr\left\{r>A\right\}=
\int^\infty_{A}\phi\left(r\right)dr\ge \int^{2A}_{A}\phi\left(r\right)dr\ge  \int^{A}_{0}\phi\left(r\right)d r=\pr\left\{r\le A\right\}=\\
1-\pr\left\{r>A\right\},
\end{gather*}
i.e. $\pr\left\{r>A\right\}\ge {1/2}$, contrary to one of the conditions assumed. That implies the statement. The estimate $r_1\ge {1\over 6\sqrt 2}\min\left(\textsf a^{-{1\over 2}},|\textsf b|^{-{1\over 2}},|\textsf c|^{-1}\right)$ follows from the equation
\begin{gather*}
4\textsf a^2 r^4+2\textsf b r^2+\textsf c r-\left(1-\nu\right)=0
\end{gather*}

\smallskip
$(ii)$ Use $\textsf b>0$, compute, get
\begin{gather*}
f''=-{\left(1-\nu\right)\over r^2}-12\textsf a^2 r^2-2\textsf b<0,\quad \text{ for $r>0$}
\end{gather*}
By Rolle theorem $f'$ can not have more than one root in $(0,+\infty)$, as claimed.

\smallskip
$(iii)$
Write the equation
\begin{gather*}
f''=-{\left(1-\nu\right)\over r^2}-12\textsf a^2 r^2-2\textsf b=0
\end{gather*}
The equation  has two positive roots at most, for it is bi-quadratic.
By Rolle theorem $f'$ can not have more than three positive roots counted with their multiplicities. Assume there exists $0<r_1<r_3$ with $f'\left(r_3\right)=0$, $f''\left(r_3\right)\le 0$. Denote $0<r_1\le r_2\le r_3$ the point at which $f$ assumes its minimum on $[r_1,r_3]$. Then $f'\left(r_2\right)=0$, $f''\left(r_2\right)\ge 0$. Part $(1)$ implies, in particular, $r_1<r_2$, for $f''\left(r_1\right)< 0$. Assume $f''\left( r_2\right)=0$. Then $f'''\left(r_2\right)=0$, for $r_2$ is a local minimum. Hence, $r_2$ is a root of $f'$ of order 3 at least. Since $f'\left(r_1\right)=0$ and  $r_1<r_2$, $f'$ has at least four roots when conted with their multiplicities. This is contrary to conclusion we made before. Thus, $f''\left(r_2\right)>0$. In particular $r_2<r_3$ and $f'$ has at least three roots $0<r_1<r_2<r_3$. Since $f$ can not have four roots this verifies $(iiia)$, $(iiib)$ along with the $r_3$--uniqueness statement in $(iii)$. Combine Rolle theorem with $(iiib)$, get $(iiic)$. 

\smallskip
$(iv)$ Equation $f'\left(r_j\right)=0$ reads
\begin{gather*}
4\textsf a^2r_j^4+2\textsf br_j^2+\textsf cr_1-1+\nu=0
\end{gather*}
Take $|\eta|<1$, use \eqref{eq:radiusdistriba}, combine with the equation in the display, write
\begin{gather*}
f\left(r_j\left(1+\eta\right)\right)- f\left(r_j\right)=\\
\left(1-\nu\right) \log \left(1+\eta\right)-\textsf a^2 r_1^4\left(1+\eta\right)^4-\textsf b r_1^2\left(1+\eta\right)^2-\textsf cr_j \left(1+\eta\right)+
\textsf a^2r_1^4+\textsf br_j^2+\textsf cr_j=\\
\left(1-\nu\right) \log \left(1+\eta\right)-\left(4\textsf a^2r_j^4+2\textsf br_j^2+\textsf cr_j\right)\eta-\left(6\textsf a^2r_j^4+\textsf b r_j^2\right)\eta^2
-4\textsf a^2r_j^4\eta^3-\textsf a^2r_j^4\eta^4=\\
\left(1-\nu\right) \log \left(1+\eta\right)-\left(1-\nu\right)\eta-\left(6\textsf a^2r_j^4+\textsf br_j^2\right)\eta^2-4\textsf a^2 r_j^4\eta^3-\textsf a^2r_j^4\eta^4,
\end{gather*}
as claimed. Combine \eqref{eq:logdensitydecay} with the standard logarithm power series expansion, get \eqref{eq:logdensitydecay2}.

\smallskip
$(v)$ The statement for $\textsf b\ge 0$ follows from part $(ii)$. Assume $\textsf b<0$. Then it is the first case in the the dichotomy conditions in $(III)$, i.e.
\be
\label{eq:dichotomy120}\max\left(\textsf a^2,\left|\textsf b\right|,\left|\textsf c\right|\right)\le \textsf A
\ee
Combine \eqref{eq:logdensitydecay2} with \eqref{eq:dichotomy120}, get the second statement in 
$(v)$.

\smallskip
$(vi)$ Use \eqref{eq:logdensitydecay2} with $|\eta|\le {1\over 8}$, combine with $\textsf b\ge 0$,
get
\begin{gather*}
{\phi\left(r_1\left(1+2\eta\right)\right)\over \phi\left(r_1\left(1+\eta\right)\right)}\le\exp W\left[-{8\over 5}\left(1-\nu\right)\eta^2-16\textsf a^2r_j^4\eta^2-\textsf br_j^2\eta^2\right]\times\\
\exp W\left[{2\over 3}\left(1-\nu\right)\eta^2+6\textsf a^2r_1^4\eta^2+\textsf br_1^2\eta^2\right]<\\
\exp W\left[-{3\over 5}\left(1-\nu\right)\eta^2-\textsf a^2r_1^4\eta^2-\textsf br_1^2\eta^2\right],
\end{gather*}
as claimed.

\smallskip
$(vii)$ In this part statement we assume $\textsf b\ge 0$. Therefore $r_1$ is the only critical point in the interval $(0,2A)$. The function  
\begin{gather*}
\Pi(t)=\int_{\phi(r)<t}\phi\left(r\right)dr,\quad 0<t<\phi(r_1)
\end{gather*}
increases monotonically and continuos. There exists unique $0<t:=\phi_{\textsf m}<\phi(r_1)$, such that $\Pi\left(\phi_{\textsf m}\right)={1\over 2}\Pi\left(\phi(r_1)\right)$. Since the function $\phi(r)$ has unique maximum at $r=r_1$,
the set $I_\textsf m=\{\phi(r)>\phi_{\textsf m}\}$ is an interval, $I_\textsf m=(r_{1,<},r_{1,>})$,  $r_{1,<}<r_1<r_{1,>}$. Assume $\left|I_\textsf m\right|\ge {r_1\over 4}$. Then $\max\left(r_1-r_{1,<},r_{1,>}-r_1\right)\ge {r_1\over 8}$. Assume first that $r_1-r_{1,<}\le r_{1,>}-r_1$. Denote $\eta={1\over 16}>0$, $r_{1,k\eta}=r_1(1+k\eta)$, $k=1,2$. Then $r_{1,2\eta}\le r_{1,>}$. Use part $(v)$, write
\begin{gather*}
\phi\left(r_{1,>}\right)\le \phi\left(r_{1,2\eta}\right)\le\phi\left(r_{1,\eta}\right)\times
\exp W\left[-{3\over 5}\left(1-\nu\right)\eta^2-\textsf a^2r_1^4\eta^2-\textsf br_1^2\eta^2\right]
\end{gather*}
Next, 
\begin{gather*}
\int_{r_1<r<r_{1,\eta}}\phi\left(r\right)dr<\int_{I_\textsf m}\phi\left(r\right)dr:=\textsf M_c=\int_{(0,+\infty)\setminus I_\textsf m}\phi\left(r\right)dr=\\
\int_{0<r<r_{1,<}}\phi\left(r\right)dr+\int_{r_{1,>}<r<+\infty}\phi\left(r\right)dr:=\textsf M_<+\textsf M_>
\end{gather*}
Assume first that $r_{1,>}<A$. Truncate the variable $r$ to the level $r=A$ in the integral $\textsf M_>$. For that use \eqref{eq:rsupport}, write
\begin{gather*}
1\gg \pr\left[\cB_\texttt r\right]\ge p_A:=\pr\left\{A<r<+\infty\right\}={\textsf M_{A}\over \textsf M_<+\textsf M_c+\textsf M_>}={\textsf M_{A}\over 2\textsf M_c},\\
\textsf M_{A}:=\int_{A<r<+\infty}\phi\left(r\right)dr,\\
\textsf M_c=\textsf M_<+\textsf M_>=\textsf M_<+\textsf M_{>,A}+\textsf M_{A}=\textsf M_<+\textsf M_{>,A}+2p_A\textsf M_c,\\
\textsf M_{>,A}:=\int_{r_{1,>}<r<A}\phi\left(r\right)dr,\\
\textsf M_c=\left(1-2p_A\right)^{-1}\left[\textsf M_<+\textsf M_{>,A}\right]<2\left[\textsf M_<+\textsf M_{>,A}\right],\\
\end{gather*}
Combine, get 
\begin{gather*}
\int_{r_1<r<r_{1,\eta}}\phi\left(r\right)dr<\left(1-2p_A\right)^{-1}\left[\textsf M_<+\textsf M_{>,A}\right]<2\left[\textsf M_<+\textsf M_{>,A}\right]
\end{gather*}
On the other hand, compare
\begin{gather*}
\phi\left(r\right)\le \phi\left(r_{1,>}\right)\le \phi\left(r_{1,2\eta}\right),\quad \text{for $r_{1,>}<r<+\infty$},\\
\phi\left(r\right)\le \phi\left(r_{1,<}\right)=\phi\left(r_{1,>}\right),\quad \text{for $0<r<r_{1,<}$},\\
\textsf M_<+\textsf M_{>,A}\le A\phi\left(r_{1,>}\right)\le A\phi\left(r_{1,\eta}\right)\times
\exp W\left[-{3\over 5}\left(1-\nu\right)\eta^2-\textsf a^2r_1^4\eta^2-\textsf br_1^2\eta^2\right]<\\
Ar_1^{-1}\exp W\left[-{3\over 5}\left(1-\nu\right)\eta^2-\textsf a^2r_1^4\eta^2-\textsf b r_1^2\eta^2\right]\int_{r_1<r<r_{1,\eta}}\phi\left(r\right)dr
\end{gather*}
Combine, use the last estimate in part $(i)$, get 
\begin{gather*}
\int_{r_1<r<r_{1,\eta}}\phi\left(r\right)dr<2Ar_1^{-1}\exp W\left[-{3\over 8}\left(1-\nu\right)\eta^2-4\textsf a^2r_1^4\eta^2-\textsf b r_1^2\eta^2\right]\int_{r_1<r<r_{1,\eta}}\phi\left(r\right)dr,\\
\exp W\left[{3\over 5}\left(1-\nu\right)\eta^2+\textsf a^2r_1^4\eta^2+\textsf br_1^2\eta^2\right]<2Ar_1^{-1}<2AW^{\textsf d},\\
 W\left[{3\over 5}\left(1-\nu\right)\eta^2+\textsf a^2r_1^4\eta^2+\textsf br_1^2\eta^2\right]<\log 2+\log A+\textsf d\log W,
\end{gather*}
where $\textsf d$ is an absolute constant. Each term in the last line is non-negative. In particular
\begin{gather*}
{3\over 8}\left(1-\nu\right)\eta^2 <W^{-1}\log 2+W^{-1}\log A+\textsf d+W^{-1}\log W
\end{gather*}
Put here $\eta={1\over 16}$, conclude that this is in contradiction with the $W\gg 1$, $\textsf d\sim 1$ setup. That verifies the statement for the case it was argued. To finish the proof consider the complementary cases. Assume $r_{1,>}\ge A$. In this case write
\begin{gather*}
1\gg \pr\left[\cB\right]\ge p_A:=\pr\left\{A<r<+\infty\right\}\ge \pr\left\{r_{1,>}<r<+\infty\right\}={\textsf M_{>}\over 2\textsf M_c},\\
\textsf M_c=\textsf M_<+\textsf M_><\textsf M_<+2p_A\textsf M_c,\\
\textsf M_c=\left(1-2p_A\right)^{-1}\textsf M_<<2\textsf M_<\\
\end{gather*}
From this point the argument goes completely similar to the case $r_{1,>}< A$. Finally, consider the last remaining case  $r_1-r_{1,<}> r_{1,>}-r_1$.
If $r_{1,<}\ge A$, then $r_{1,>}> A$ and the argument goes completely similar to the case $r_{1,>}< A$. If $r_{1,<}< A$, the argument goes similar to the case $r_{1,>}< A$, but it is shorter for there is no need to truncate the variable $r$. That finishes $(vii)$. 

\smallskip
$(viii)$ Use the notations from the part $(vii)$ proof, write
\begin{gather*}
\pr\left\{r_{1,<}<r<r_{1,>}\right\}={\textsf M_{c}\over \textsf M_<+\textsf M_c+\textsf M_>}={\textsf M_{c}\over 2\textsf M_c}={1\over 2}
\end{gather*}
Use \eqref{eq:dichotomy2}, conclude there exists $\check r\notin \cB_\texttt r$, $r_{1,<}<\check r<r_{1,>}$, which implies 
\be
\nn \textsf a^2\check r^4,\textsf b\check r^2\le 4A^2\qquad\qquad
\ee
Combine this with the assumption $\left|I_\textsf m\right|< {r_1\over 4}$, get
\begin{gather*}
r_1\le \check r+{r_1\over 9},\quad r_1\le {9\over 8}\check r,\\
\max\left[\textsf a^2r_1^4,|\textsf b|r_1^2\right]<2\max\left[\textsf a^2\check r^4,|\textsf b|\check r^2\right]\le 8A^2,
\end{gather*}
as claimed.

\smallskip
$(ix)$ The argument for the first statement in $(ix)$ is completely similar to the one for the first statement in $(vii)$. To prove the second statement use \eqref{eq:logdensitydecay}, combine with \eqref{eq:arbound1}, get 
\begin{gather*}
f\left(r_1\left(1+\eta\right)\right)-f\left(r_1\left(1+\eta+{\kappa\over \sqrt W}\right)\right)\le \\
\left(1-\nu\right) \log \left(1+\eta\right)-\left(1-\nu\right) \log \left(1+\eta+{\kappa\over \sqrt W}\right)+CA^2\times{\kappa\over \sqrt W}\le\\
2CA^2\times{\kappa\over \sqrt W},
\end{gather*}
where $C$ is an absolute constant. As in $(ix)$ statement, denote for convenience $r_{>;\eta'}=r_1\left(1+\eta'\right)$. Take arbitrary integer $k\ge 1$, set in the display $\kappa=\kappa_k=C^{-1}A^{-2}k^{-2}$, $\eta=\eta_j=j\kappa_k$, $j=0,...,k$, write
\begin{gather*}
\textsf M_{>;j}:=\int_{r_{>;\eta_{j}}<r<r_{>;\eta_{j+1}}}\phi\left(r\right)dr\le e^{{1\over k^2}}\int_{r_{>;\eta_{j+1}}<r<r_{>;\eta_{j+2}}}\phi\left(r\right)dr:=\textsf M_{>;j+1},\\
\textsf M_{>;0}\le e^{{j\over k^2}}\textsf M_{>;j},\quad j=0,...,k-1,\\
\int_{0<r<+\infty}\phi\left(r\right)dr\ge \int_{0<r<r_{>;\eta_{k}}}\phi\left(r\right)dr=\sum_{0\le j\le k-1}\textsf M_{>;j}\ge \\
\textsf M_{>;0}\times\sum_{0\le j\le k-1} e^{-{j\over k^2}}\ge\textsf M_{>;0}\times\sum_{0\le j\le k-1} 1-{j\over k^2}>(k-1)\times \textsf M_{>;0}
\end{gather*}
Recall that we assume in $(ix)$ that
\begin{gather*}
\textsf M_{>,c}:=\int_{r_1<r<r_{1,>}}\phi\left(r\right)dr\ge \int_{r_{1,<}<r<r_1}\phi\left(r\right)dr:=\textsf M_{<,c}
\end{gather*}
Use the proof of part $(viii)$, combine with the last display, get
\begin{gather*}
\textsf M_c+\textsf M_<+\textsf M_>=\int_{0<r<+\infty}\phi\left(r\right)dr,\\
\textsf M_<+\textsf M_>=\int_{0<r<r_{1,<}}\phi\left(r\right)dr+\int_{r_{1,>}<r<+\infty}\phi\left(r\right)dr,
\textsf M_c=\textsf M_<+\textsf M_>,\\
\textsf M_{>,c}\ge {1\over 4}\int_{0<r<+\infty}\phi\left(r\right)dr
\end{gather*}
Combine, get
\begin{gather*}
\textsf M_{>,c}\ge {1\over 4}\int_{0<r<+\infty}\phi\left(r\right)dr>{k-1\over 4}\times \textsf M_{>;0}
\end{gather*}
Take here $k=21$, write
\begin{gather*}
\int_{r_1<r<r_{>;\eta_{1}}}\phi\left(r\right)dr=\textsf M_{>;0}<{1\over 5}\textsf M_{>,c}={1\over 5}\int_{r_1<r<r_{1,>}}\phi\left(r\right)dr
\end{gather*}
That implies
\begin{gather*}
r_{>|{1\over 5}}>r_{>;\eta_{1}}=r_1\left(1+{c_1\over\sqrt W}\right), 
\end{gather*}
where $c_1:=\kappa_{21}$ is an absolute constant. Note that
\begin{gather*}
\int_{r_{>|{1\over 5}}<r<r_{1,>}}\phi\left(r\right)dr={4\over 5}\textsf M_{>,c}
\end{gather*}
That allows to use the above arguments with $r_{>|{1\over 5}}$ in the role of $r_1$. That leads to 
\begin{gather*}
\int_{r_{>|{1\over 5}}<r<r_{1,>}}\phi\left(r\right)dr\ge {1\over 5}\int_{0<r<+\infty}\phi\left(r\right)dr>{k-1\over 5}\times \int_{r_{>|{1\over 5}}<r<r_{>|{1\over 5}}\left(1+\kappa_k\right)}\phi\left(r\right)dr
\end{gather*}
One can see that $k=21$ is again a right choice. Indeed, write
\begin{gather*}
\int_{r_{>|{1\over 5}}<r<r_{>|{1\over 5}}\left(1+\kappa_k\right)}\phi\left(r\right)dr<{1\over 4}\int_{r_{>|{1\over 5}}<r<r_{1,>}}\phi\left(r\right)dr\le \\
{1\over 5}\int_{r_1<r<r_{1,>}}\phi\left(r\right)dr
\end{gather*}
That implies
\begin{gather*}
r_{>|{2\over 5}}>r_{>|{1\over 5}}\left(1+{c_1\over\sqrt W}\right), 
\end{gather*}
where $c_2$ is an absolute constant. The rest of the proof for $(ix)$ is similar.

\smallskip
$(x)$ The proof in case $\textsf b<0$ is similar to the one for $\textsf b$ and actually shorter. The reason for the shortcut is that in this case equation \eqref{eq:dichotomy120} holds. In particular,
\be
\label{eq:dichotomy121}\max\left(\textsf a^2r_1^4,\left|\textsf b\right|r_1^2\right)\le A_1^2
\ee
with another absolute constant $A_1$. This estiame is the only one where we used the function $\phi(r)$ decay rate on its two monotonicity intervals from part $(iv)$. With ${eq:dichotomy121}$ in hand the proof goes completely similar to the one in $(ix)$ and we skip it.

\smallskip
$(xi)$ The argument is completely similar to the one in $(x)$ and we skip it.
\end{proof}
Now we prove Proposition~\ref{prop:superexplogvar}.
\begin{proof}[Proof of Proposition ~\ref{prop:superexplogvar}] Assume that conditions $(I)-(III)$ hold. 
Use parts $(ix)$, $(x)$, $(xi)$ from Proposition~\ref{prop:elementary2g}, conclude that in any event there exist five points, which we denote as $r_{|j}<r_{<|j+1}$, $j=0,1,2,3$, such that the following equations hold
\begin{gather*}
\int_{r_{|0}<r<r_{|4}}\phi\left(r\right)dr\ge {1\over 8}\int_{0<r<+\infty}\phi\left(r\right)dr,\\
\int_{r_{|j}<r<r_{|j+1}}\phi\left(r\right)dr\ge {1\over 5}\int_{r_{|0}<r<r_{|4}}\phi\left(r\right)dr,\quad j=0,1,2,3\\
r_{|j+1}>r_{|j}\left(1+{c\over\sqrt W}\right), 
\end{gather*}
where $c>0$ is an absolute constant. Combine first two equations, get
\begin{gather*}
\pr\left\{r_{|j}<r<r_{|j+1}\right\}\ge{1\over 40},\quad j=0,1,2,3
\end{gather*}
Take arbitrary constant $\bar r\ge 0$. Find smallest $j$ such that $\bar r\le r_{|j}$, evaluate
\begin{gather*}
\expc\left [\log r-\log \bar r\right]^2\ge\\
{1\over 40}\begin{cases}
\left[\log r_{|2}-\log r_{|1}\right]^2\ge \left[\log \left(1+{c\over\sqrt W}\right)\right]^2\ge {c^2\over 2W},\quad\text{if $j=0$},\\
\left[\log r_{|3}-\log r_{|2}\right]^2\ge \left[\log \left(1+{c\over\sqrt W}\right)\right]^2\ge {c^2\over 2W},\quad\text{if $j=1$},\\
\left[\log r_{|4}-\log r_{|3}\right]^2\ge \left[\log \left(1+{c\over\sqrt W}\right)\right]^2\ge {c^2\over 2W},\quad\text{if $j=2$},\\
\left[\log r_{|1}-\log r_{|0}\right]^2\ge \left[\log \left(1+{c\over\sqrt W}\right)\right]^2\ge {c^2\over 2W},\quad\text{if $j=3$},\\
\end{cases},
\end{gather*}
That implies the first statement. 

To verify the second statement use Proposition~\ref{prop:elementary2g}, $(i)$, write
\be\label{eq:rlowerest}
r_1\ge c_0\min\left(\textsf a^{-{1\over 2}},|\textsf b|^{-{1\over 2}},|\textsf c|^{-1}\right)\ge cW^{-\textsf d}
\ee
with absolute constants $c,\textsf d$. Next we need to use the function $\phi(r)$ decay rate on its monotonicity interval $(0,r_1)$. For that reason we again use 
the dichotomy conditions in $(III)$ and consider separately the cases $\textsf b\ge 0$ and $\textsf b<0$. 

Assume first that $\textsf b\ge 0$. In this case we use \eqref{eq:logdensitydecay2} from Proposition~\ref{prop:elementary2g}, $(iv)$ and write
\begin{gather*}
\phi\left(r_{1,-\eta}\right)\le\phi\left(r_{1}\right)\times
\exp \left[-W{3\over 5}\left(1-\nu\right)\eta^2\right], \quad 0<\eta<{1\over 8},\\
\phi\left(r\right)\le \phi\left(r_{1,-\eta}\right) \quad 0<r<\eta<{1\over 8}
\end{gather*}
On the other hand in this case
\be
\nn \max\left(\textsf a^2r_1^4,\left|\textsf b\right|r_1^2\right)\le 8A
\ee
Use \eqref{eq:logdensitydecay2} left inequality, get 
\begin{gather*}
\phi\left(r_1\left(1+{\eta\over \sqrt W}\right)\right)\ge \\
\exp \left[-{2\over 3}\left(1-\nu\right)\eta^2-6\textsf a^2r_j^4\left(1+{3\eta\over 2}\eta+{\eta^2\over 6}\right)\eta^2-\textsf br_j^2\eta^2\right]\times 
\phi\left(r_1\right)\ge \\
e^{-A_1\eta^2}\times 
\phi\left(r_1\right), \quad -1<\eta<1
\end{gather*}
with absolute constant $A_1$. Combine, get
\begin{gather*}
\pr\left\{0<r<r_{1,-{1\over 8}}\right\}={\textsf M_{(0,r_{1,-{1\over 8}})}\over \textsf M_{(0,+\infty)}}\le {\textsf M_{(0,r_{1,-{1\over 8}})}\over \textsf M_{(r_{1,-{1\over \sqrt W}}, r_{1,{1\over \sqrt W}})}},\\
\textsf M_{(0,r_{1,-{1\over 8}})}\le r_1\times \phi\left(r_{1}\right)\times\exp \left[-W{3\over 5\times 64}\left(1-\nu\right)\right]\le \\
2A\times \phi\left(r_{1}\right)\times\exp \left[-W{3\over 5\times 64}\left(1-\nu\right)\right],\\
\textsf M_{(r_{1,-{1\over \sqrt W}}, r_{1,{1\over \sqrt W}})}\ge {2\over \sqrt W}\times e^{-A_1}\times 
\phi\left(r_1\right),\\
\pr\left\{0<r<r_{1,-{1\over 8}}\right\}\le A\sqrt W\times\exp \left[-W{3\over 5\times 64}\left(1-\nu\right)+A_1\right]\le e^{-cW},
\end{gather*}
where $c$ is absolute constant. Combine with \eqref{eq:rlowerest}, get the statement in case $\textsf b\ge 0$.

Assume $\textsf b< 0$. Then it is the first case in the the dichotomy conditions in $(III)$. In particular
\be
\label{eq:dichotomy122}\textsf a^2,\left|\textsf b\right|\le \textsf A
\ee
Write
\begin{gather*}
f''(r)=-{\left(1-\nu\right)\over r^2}-12\textsf a^2 r^2-2\textsf b<-{\left(1-\nu\right)\over r^2}+\textsf A<\\
-{\left(1-\nu\right)\over 2r^2},\quad \text{ for $0<r<c$},\\
f'(r)<-c_1,\quad \text{ for $0<r<c$},\\
f(r)<f(r_1)-c_1(r_1-r),\quad \text{ for $0<r<\min\left(r_1,c\right)$}
\end{gather*}
with absolute constants $c,c_1$. With \eqref{eq:dichotomy122} and the last inequality in hand the argument goes completely similar to the one in case $\textsf b< 0$ and we skip it.

That finishes the proof of Proposition ~\ref{prop:superexplogvar}.
\end{proof}

\section{Proof of Theorem B.}\label{sec:Theorem B}
The proof of Theorem B basically consists of application of Propositions~\ref{prop:condecoupling1},\ref{prop:condecoupling},\ref{prop:superexplogvar} to each site $n$ factor in 
the Green function edge--to--edge matrix vector action
\be\label{eq:Greenedgetoedge1}
G_{[1,N]}(1;N)g= (-1)^{N-1}U_1^{-1}\times T_1\times ...\times U^{-1}_{N-1}\times T_{N-1}\times U_N^{-1}g
\ee
It is convenient to split the proof of Theorem B into two proposition and a closing argument.The main point of such presentation is that each proposition statement introduces notations and desribes in a transparent way what is the target. The proofs are easy because we use detailed statements from of Propositions~\ref{prop:condecoupling1},\ref{prop:condecoupling},\ref{prop:superexplogvar}.
The first proposition below describes repeated application of Proposition~\ref{prop:condecoupling1} to the factors in \eqref{eq:Greenedgetoedge1}.
\begin{prop}\label{prop:sitescondecoupling1} $(0)$ Fix arbitrary $U_1,U_2, ..., U_N$. In \eqref{eq:Greenedgetoedge1} redenote for convenience $g:=g_{N}$. Apply Propositions~\ref{prop:condecoupling1} with $n=N$, denote $Q_{1,N}$ the orthogonal matrix $Q_1$ in the equation \eqref{eq:vectorchange0}, $Q_{2,N-1}$ the orthogonal matrix $Q_2$ defined in the part $(2)$ of the proposition, $X_{1,N-1|T_{N-1}}$ the vector in the equation \eqref{eq:siteindex}, $g_{N-1|T_{N-1}}$ the vector $h$, defined in the equation \eqref{eq:vectorcond1}. Apply Propositions~\ref{prop:condecoupling1} with $n=N-1$, $g_{N-1|T_{N-1}}$ in the role of $g$, denote $X_{1,N-2|T_{N-2},T_{N-1}}$ the vector from \eqref{eq:siteindex}, $g_{N-1|T_{N-2},T_{N-1}}$ the vector from \eqref{eq:vectorcond1}. Repeat the application, get the vectors $X_{1,n-1|T_{n-1},...,T_{N-1}}$, $g_{1,n-1|T_{n-1},...,T_{N-1}}$.

The following equation holds
\be\nn
U_1^{-1}\times T_1\times ...\times U^{-1}_{N-1}\times T_{N-1}\times U_N^{-1}g=\qquad\qquad\qquad\qquad\\
\label{eq:Greenedgetoedgesplit}\\
\nn\prod_{2\le n\le N} \left\|U_n^{-1}g_{1,n-1|T_{n-1},...,T_{N-1}}\right\|^{-1}\times\left\|X_{1,n-1|T_{n-1},...,T_{N-1}}\right\|^{-1}\times U_1^{-1}g_{1,1|T_{1},...,T_{N-1}}
\ee

\smallskip
$(1)$ The random variables $\left\|X_{1,n-1|T_{n-1},...,T_{N-1}}\right\|$, $n=2,...,N$ are independent when conditioned on the random matrices  $U_2, ..., U_N$ and on the random vectors 
$g_{1,n-1|T_{n-1},...,T_{N-1}}$, $n=2,...,N$.

\end{prop}
\begin{proof} Rewrite the equations \eqref{eq:schurfactorvectoraction},
\eqref{eq:vectorcond} \eqref{eq:vectorcond1} using the above notations :
\be\nn
\textsf v_{g}\left(T_{n-1}\big|\right):=T_{n-1}\times U_n^{-1}g_{1,n-1|T_{n-1},...,T_{N-1}}
\ee
\be\nn
\left\|\textsf v_n\left(T_{n-1}\big|\right)\right\|^{-1}\textsf v_n\left(T_{n-1}\big|\right)=g_{1,n-2|T_{n-2},...,T_{N-1}}=
\ee
\be
\quad\left\|X_{1,n-1|T_{n-1},...,T_{N-1}}\right\|^{-1}Q_{2,n-1}X_{1,n-1|T_{n-1},...,T_{N-1}}\nn
\ee
Put this in \eqref{eq:Greenedgetoedge1}, get the equation in part $(0)$. 

To verify part $(1)$, recall 
that he conditional distribution of the random vector $X_{1,n-1|T_{n-1},...,T_{N-1}}$ is as follows
\be\nn\nn d\pr\left(X_{1,n-1|T_{n-1},...,T_{N-1}}\big|\right)=\qquad\qquad\\
\label{eq:Xdistribution}\\
\nn\mu e^{-{W\over 2}\phi\left(X_{1,n-1|T_{n-1},...,T_{N-1}}\big|\right)}\times dX_{1,n-1|T_{n-1},...,T_{N-1}},
\ee
where $\mu$ is the normalizing factor and the function $\phi\left(X_{1,n-1|T_{n-1},...,T_{N-1}}\big|\right)$ depends on $U_{n-1},U_{n}$ and the vector $g_{1,n|T_{n},...,T_{N-1}}$ only, see the statement in part $(3)$ of Propositions~\ref{prop:condecoupling1} for more details. That implies $(1)$.
\end{proof}
Next we apply Propositions~\ref{prop:condecoupling} to each factor $\left\|X_{1,n-1|T_{n-1},...,T_{N-1}}\right\|^{-1}$ in \eqref{eq:Greenedgetoedgesplit}.
\begin{prop}\label{prop:condecoupling2} Use the notations in Proposition~\ref{prop:sitescondecoupling1}.

\smallskip
$(0)$ Use the spherical variables, write 
\begin{gather*}
X_{1,n-1|T_{n-1},...,T_{N-1}}=r_{|n-1}\textsf x_{|n-1},\\
r_{|n-1}=\left\|X_{1,n-1|T_{n-1},...,T_{N-1}}\right\|,\quad \textsf x_{|n-1}\in\mathbb S^{W-1}
\end{gather*}
Then 
\be
g_{1,n-2|T_{n-2},...,T_{N-1}}=Q_{2,n-1}\textsf x_{|n-1},
\label{eq:sphericaln}
\ee
\be
\log\left\|G_{[1,N]}(1;N)g\right\|=L+\sum_{2\le n\le N} \log r_{|n-1},
\label{eq:Greenedgetoedgesplitf}
\ee
where $L$ does not depend on the variables $r_{|n-1}$. The variables $r_{|n-1}$,  $n=2,...,N$ are independent when conditioned on the random matrices  $U_2, ..., U_N$ and on the random vectors 
$\textsf x_{|n-1}$, $n=2,...,N$.

\smallskip
$(1)$ The conditional distribution $d\pr\left(r_{|n-1}\big|\right)$ is as follows
\begin{gather*}
d\pr\left(r_{|n-1}\big|\right)=\mu\Big(U_{n-1},U_{n},\textsf x_{|n-1},\textsf x_{|n}\Big)\times\\
r_{|n-1}^{W-1}e^{-{W\over 2}\left(\textsf a_{|n-1}^2 r_{|n-1}^4+\textsf b_{|n-1} r_{|n-1}^2+\textsf c_{|n-1} r_{|n-1}\right)}dr_{|n-1},
\end{gather*}
where 
\begin{gather*}
\text{$\textsf a_{|n-1}^2,\textsf b_{|n-1},\textsf c_{|n-1}$ are constants depending $U_{n-1},U_{n},\textsf x_{|n-1},\textsf x_{|n-1}$},\\
\text{$\mu\Big(U_{n-1},U_{n},\textsf x_{|n-1},,\textsf x_{|n}\Big)$ is the normalizing factor, $n=2,...,N$}
\end{gather*}

\smallskip   
$(2)$  There exists a set $\widehat{\cB}\in \mathcal F\left(U_2, ..., U_N\right)$,  such that
\be\label{eq:inverscrude}
\pr_N\left[\widehat{\cB}\right]\le CNW^{D-1},
\ee
where $D\gg 1$ is an absolute constant, and  for each $\left(U_2, ..., U_N\right)\notin \widehat{\cB}$ there exist sets $\cB_{|n-1,A}\in \mathcal F\left(T_{n-1},\textsf x_{|n}\right)$, depending on  $\left(U_2, ..., U_N\right)$, with 
\be\label{eq:Tnlimitations1}
\pr_N\cB_{|,n-1,A}\le C_1e^{-c_1A^2W}
\ee
and  for each $\left(T_{n-1},\textsf x_{|n-1}\right)\notin \cB_{|,n-1,A}$ there exist sets $\cB_{|\texttt r,n-1,A}\in \mathcal F\left(r_{|n-1}\right)$, depending on  $\left(T_{n-1},\textsf x_{|n}\right)$, with 
\be\label{eq:Tnlimitations2}
\pr_N\cB_{|\texttt r,n-1,A}\le C_1e^{-c_1A^2W}
\ee
where $A,C_1,c_1$ are absolute constants, $n=2,...,N$, such that for any $n$ and
\begin{gather*}
\text{for any $\left(U_2, ..., U_N\right)\notin \widehat{\cB}$},\\
\text{and any $\left(T_{n-1},\textsf x_{|n}\right)\notin \cB_{|n-1,A}$, $\left(T_{n-2},\textsf x_{|n-1}\right)\notin \cB_{|n-2,A}$},\\
\text{and any $r_{|n-1}\notin \cB_{|\texttt r,n-1,A}$, $r_{|n-2}\notin \cB_{|\texttt r,n-2,A}$}
\end{gather*}
conditions $(I)-(II)$ in Proposition~\ref{prop:superexplogvar} hold for $d\pr\left(r_{|n-1}\big|\right)$,$d\pr\left(r_{|n-2}\big|\right)$. Moreover, 
either $d\pr\left(r_{|n-1}\big|\right)$ obeys Proposition~\ref{prop:superexplogvar} condition $(III)$ dichotomy first case, or $d\pr\left(r_{|n-2}\big|\right)$ obeys Proposition~\ref{prop:superexplogvar} condition $(III)$  dichotomy second case.
Finally, the variables $\mathcal I_{\IR^W\times\mathbb S^{W-1}\setminus \cB_{|n-1,A}}\left(T_{n-1}\right)\times r_{|n-1}$,  $n=2,...,N$ are independent, when conditioned on the random matrices  $U_2, ..., U_N$. Here $\mathcal I_{\cA}$ stands for the indicator of the set $\cA$.
\end{prop}
\begin{proof} Part $(0)$ follows straightforward from Proposition~\ref{prop:sitescondecoupling1}. Part $(1)$ follows straightforward from part $(0)$ in Proposition~\ref{prop:condecoupling}. To verify $(2)$ recall that part $(1)$ in Proposition~\ref{prop:condecoupling} reads  in particular:

there exists a set $\widehat{\cB}_{U,n-2,n-1,D}\in \mathcal F\left(U_{n-2},U_{n-1}\right)$,  such that
\begin{gather*}
\pr_N\left[\widehat{\cB}_{U,n-1,n,D,A}\right]\le CW^{D-1},
\end{gather*}
and a set $\cB_{S,n-1,A}\in \mathcal F\left(S\right)$ with 
\begin{gather*}
\pr_N\cB_{S,n-1,A}\le C_1e^{-c_1A^2W}
\end{gather*}
such that for any $\left(U_{n-1},U_{n}\right)\notin \widehat{\cB}_{U,n-1,n,D,A}$, $S\notin \cB_{S,n-1,A}$ and any $\textsf x$ holds 
\begin{gather*}
\max\left(\textsf a^2\left(\textsf x\right),\left|\textsf b\left(\textsf x\right)\right|,\left|\textsf c\left(\textsf x\right)\right|\right)\le 2A^2W^{2D}
\end{gather*}
Recall also that from its definition $S=Q_{2,n-1}T_{n-1}Q_{1,n-1}$. Denote 
\begin{gather*}
\widehat{\cB}=\bigcup_{1\le n\le N}\widehat{\cB}_{U,n-1,n,D,A},\\
\cB_{|\texttt T,n-1,A}=\left\{T_{n-1}:Q_{2,n-1}T_{n-1}Q_{1,n-1}\in \cB_{S,n-1,A}\right\},\\
\cB_{|n-1,A}=\cB_{|\texttt T,n-1,A}\cup \left\{\textsf x_{|n-1}: \textsf x_{|n-1}\in\cB_{\texttt x,n-1,n,A}\right\},\\
\cB_{|n-1,A}=\cB_{|\texttt T,n-1,A}=\left\{r_{|n-1}: r_{|n-1}\in \cB_{\texttt r,\textsf x,n-1,n,A}\right\},
\end{gather*}
where $\cB_{\texttt x,n-1,n,A}\in \mathcal F\left(\textsf x\right)$, $\cB_{\texttt r,\textsf x,n-1,n,A}\in \mathcal F\left(r\right)$ are from the statement in part $(3)$ of Proposition~\ref{prop:condecoupling}. Equation \eqref{eq:inverscrude} holds by subadditivity. 
holds becasue of Gaussian random matrix distribution orthogonal invariance.  Equations \eqref{eq:Tnlimitations1}, \eqref{eq:Tnlimitations2} follow from part $(3)$ in Proposition~\ref{prop:condecoupling}. The statement addressing conditions $(I)-(II)$ in Proposition~\ref{prop:superexplogvar} follow from parts $(2)$, $(3)$ in Proposition~\ref{prop:condecoupling}. The statement addressing the dichotomy follows from the last statement in part $(5)$ of Proposition~\ref{prop:condecoupling}.

Finally, the variables $T_{n}$, $n=1,...,N-1$ are independent, when conditioned on $U_{n}$, $n=1,...$. That implies the last statemnt.
\end{proof}
The proof of Theorem B consists of a short argument based on the following corollary, which we use also to prove Theorem C. To state the corollary we need to introduce yet another collection of sets. Here is the discription. In the proof of the corollary we use both statements in Proposition~\ref{prop:superexplogvar}, applied to the super-exponential density from part $(1)$ in Proposition~\ref{prop:condecoupling2}, i.e. to
\begin{gather*}
\text{$a_{|n-1}^2,\textsf b_{|n-1},\textsf c_{|n-1}$ are constants depending $U_{n-1},U_{n},\textsf x_{|n-1},\textsf x_{|n-1}$}
\end{gather*}
in the role of $\textsf a,\textsf b,\textsf c$, provided conditions in Proposition~\ref{prop:condecoupling2} part $(2)$ hold. Given $\textsf a,\textsf b,\textsf c$, the second statement in Proposition~\ref{prop:superexplogvar}  defines a set $\cB_{\phi}$, such that
\be
\nn \pr \cB_{\phi}\le Ce^{-cW}\qquad\qquad\qquad\qquad\\
\label{eq:radiuslowerb1}\\
\nn r\ge c_1W^{-\textsf d},\quad \text{for any $r\in (0,+\infty)\setminus \cB_{\phi}$},
\ee
see \eqref{eq:radiuslowerb}. Given $U_{n-1},U_{n},\textsf x_{|n-1},\textsf x_{|n-1}$, which obey conditions in Proposition~\ref{prop:condecoupling2} part $(2)$, redenote $\cB_{\phi}$ as $\cB_{|\texttt r,n-1}$. These are new sets we use in  sets we use in the corollary below.

\begin{cor}\label{cor:logrindep}
Use the notations in Proposition~\ref{prop:condecoupling2} part $(2)$. 
Take
\begin{gather*}
\left(U_2, ..., U_N\right)\notin \widehat{\cB}
\end{gather*}
Denote
\be
\nn\eta_{k,|} =\mathcal I_{\IR^W\times\mathbb S^{W-1}\setminus \cB_{|N-2k,A}}\left(T_{n-2k},\textsf x_{|n-2k}\right)\times \qquad\qquad\\
\nn\mathcal I_{\IR^W\times\mathbb S^{W-1}\setminus \cB_{|n-2k+1,A}}\left(T_{n-2k+1},\textsf x_{|n-1}\right)\times\qquad\qquad\\
\label{eq:logrk}
\quad\mathcal I_{(0,+\infty)\setminus \cB_{|\texttt r,n-2k,A}}\left(r_{n-2k}\right)\times \mathcal I_{(0,+\infty)\setminus \cB_{|\texttt r,n-2k+1,A}}\left(r_{n-2k+1}\right)\times\\
\nn\mathcal I_{(0,+\infty)\setminus \cB_{|\texttt r,n-2k}}\left(r_{n-2k}\right)\times \mathcal I_{(0,+\infty)\setminus \cB_{|\texttt r,n-2k+1}}\left(r_{n-2k+1}\right)\times\\
\nn\left[\log r_{|N-2k}+\log r_{|N-2k+1}\right],\quad k=1,...\qquad\qquad
\ee
The variables $\eta_{k,|}$ are independent,
\be
\nn\Var \left[\eta_{k,|}\right]\ge {c\over W} \qquad\qquad\\
\label{eq:Varlogr}\\
\nn\left|\eta_{k,|}\right|\le C\log W,\qquad\qquad
\ee
where $c,C$ are absolute constants.
\end{cor}
\begin{proof}
Combine Proposition~\ref{prop:condecoupling2} with Proposition~\ref{prop:superexplogvar}, get both estimates.
\end{proof}
\begin{proof}[Proof of Theorem B] Use \eqref{eq:Greenedgetoedgesplitf} from Proposition~\ref{prop:condecoupling2}, write
\be
\nn \textsf L_{N,|}:=\log\left\|G_{[1,N]}(1;N)g\right\|=L+\sum_{2\le n\le N-1} \log r_{|n-1}=\\
\label{eq:Greenedgetoedgesplitfco nd}\\
\nn \hat L+\sum_{2\le 2k\le N-1} \eta_{k,|}\qquad\qquad\qquad\qquad
\ee
Here $\textsf L_{N,|}$ stands for the random variable 
$$\log\left\|G_{[1,N]}(1;N)g\right\|,\quad\text{conditioned on $\left(U_2, ..., U_N,\textsf x_{1,|},...,\textsf x_{N-1,|}\right)$},$$ 
and also on $r_{1,|}$ if $N$ is even,  $\hat L$ does not depend on the variables $\eta_{|k}$, see Corollary~\ref{cor:logrindep}. Use Corollary~\ref{cor:logrindep}, get
\be
\Var\left[\textsf L_{N,|}\right]=\sum_{2\le 2k\le N} \Var\left[\eta_{k,|}\right]\ge c_1N W^{-1}, 
\label{eq:GreenedgetoedgesplitfcondVar}
\ee
where $c_1$ is an absolute constant. It follows from Proposition~\ref{prop:condecoupling2} that the variables $\textsf L_{N,|}$ is conditioned on, run in the set with $\pr_N\ge 1-N\times W^{-3D}$ where $D\gg 1$ is an absolute constant. Take $ N\sim W^{cD}$, with an absolute constant $0<c\ll 1$. Then the probability estimate results 
$\pr_N\ge 1-W^{-2D}$. It follows also from Proposition~\ref{prop:condecoupling2} that each variable $r_{n,|}$ in the setting of Corollary~\ref{cor:logrindep}, runs in the set with $\pr_N\ge 1-Ce^{-cW}$ with absolute constants $c,C$. Combine this with \eqref{eq:GreenedgetoedgesplitfcondVar} and \eqref{eq:Greenedgetoedgesplitfco nd}, conclude
\be
\Var_\cG\left[\log\left\|G_{[1,N]}(1;N)g\right\|\right]\ge c_1N W^{-1}, 
\label{eq:GreenedgetoedgeVar}
\ee
where $\pr_N\cG\ge 1-W^{-D}$ and $\Var_\cG$ stands for the conditional variance against the set $\cG$. Note that it is exactly the scales setting in theorem B. That implies the statement in Theorem B.
\end{proof}

\section{Proofs of Theorem C.}\label{sec:thmC}
We derive Theorem C from the following lemma. The proof is straightforward, but discuss it for completeness.
\begin{lemma}\label{lem:greenorlowmoments} Take a random variable $\eta$. Assume
\be\label{eq:Bernsteicond}
\left|\expc\left[\eta-\mathfrak a\right]^{n}\right| \le n!
\sigma^2 H^{n-2},\quad \text{for any $n=2,...$},
\ee
where we denote for convenience $\expc[\eta]=\mathfrak a$, $\Var \xi=\sigma^2$ and $H>0$ is a constant. Then for $|t|\le c H^{-1}$, with absolute constant $0<c\ll 1$ holds
\begin{gather*}
\Big[\expc \big[e^{t\eta}\big]\Big]^2\le\expc \big[e^{2t\eta}\big]\times \Big(1-{t^2\sigma^2\over 4}\Big)
\end{gather*}
\end{lemma}
\begin{proof} 
Write the series expansion, combine with equation \eqref{eq:Bernsteicond}, get
\begin{gather*}
\expc \big[e^{t(\eta-\mathfrak a)}\big]=\sum_{p=0}^\infty{t^p\over p!}\expc\big[(\eta-\mathfrak a)^{p}\big]
\le 1+{t^2\sigma^2\over 2}+\sum_{p=3}^\infty t^p\sigma^2H^{p-2}\\
 \le 1+{t^2\sigma^2\over 2}+2Ht^{3}\sigma^2\le 1+(1+c){t^2\sigma^2\over 2}
\end{gather*}
provided $|t|\le c\min\big(H^{-1},\sigma^{-1}\big)$. 
Similarly,
\begin{gather*}
\expc \big[e^{2t(\eta-\mathfrak a)}\big]\ge 1+2(1-c)t^2\sigma^2 
\end{gather*}
Combine, get
\begin{gather*}
\Big[\expc \big[e^{t(\eta-\mathfrak a)}\big]\Big]^2\le\expc \big[e^{2t(\eta-\mathfrak a)}\big]\times{(1+(1+c){t^2\sigma^2\over 2})^2\over (1+2(1-c)t^2\sigma^2)}\le \\
\expc \big[e^{2t(\eta-\mathfrak a)}\big]\times{(1+(1+c){t^2\sigma^2\over 2})^2\over (1+(1+c)t^2\sigma^2)}\le \expc \big[e^{2t(\eta-\mathfrak a)}\big]\times\big(1-{t^2\sigma^2\over 4}\big),
\end{gather*}
Combine with
\begin{gather*}
\Big[\expc \big[e^{t(\eta-\mathfrak a)}\big]\Big]^2=\Big[\expc \big[e^{t\eta}\big]\Big]^2\times e^{-2t\mathfrak a}\\
 \expc \big[e^{2t(\eta-\mathfrak a)}\big]=\expc \big[e^{2t\eta}\big]\times e^{-2t\mathfrak a},
\end{gather*}
get the statement.
\end{proof}
\begin{cor}\label{cor:classicalexpest} Take independent random variables $\eta_k$, $k=1,...,N$. Denote $\sigma^2_k=Var \eta_k$. Assume that each $\eta_k$ obeys condition \eqref{eq:Bernsteicond} in Lemma~\ref{lem:greenorlowmoments} with the same $H$ for all $k$ and with $\sigma_k$ in the role of $\sigma$. Denote 
\begin{gather*}
\textsf L=\sum_{1\le k\le N}\eta_k
\end{gather*}
Then for $|t|\le c H^{-1}$
\be\label{eq:Bernsteinest}
\Big[\expc \big[e^{t\textsf L}\big]\Big]^2\le\expc \big[e^{2t\textsf L}\big]\times\prod_{1 \leq k \leq N}\Big(1-{t^2\sigma_k^2\over 4}\Big)
\ee

\end{cor}
\begin{proof}
Use Lemma~\ref{lem:greenorlowmoments}, combine with independence, get \eqref{eq:Bernsteinest}.
\end{proof}
Finally, to prove Theorem C we need the following well-known form of the Wegner estimate
\be\label{eq:Wegner}
\expc \left[\left\|G_{[1,N]}\left(1;N\right)\right\|^t\right]\le C_1W^{C_2}, \quad\text{ for any $0<t<{1\over 2}$}
\ee
with absolute constants $C_1,C_2$. For reference see for instance ~\cite{Sch-09-Eigenvector}, Lemma W and equation $(1.7)$.
\begin{proof}[Proof of Theorem C] Fix
\begin{gather*}
\left(U_2, ..., U_N\right)\notin \widehat{\cB}
\end{gather*}
To simplify the notations us the $\pr$ and $\expc$ for the \emph{conditional probability and conditional expectation, at $\left(U_2, ..., U_N\right)$}.
Use \eqref{eq:Greenedgetoedgesplitfco nd} from the proof of Theorem B, write
\be\label{eq:Bernexpest}
\expc \left[e^{t\textsf L_{N,|}}\big|\right]=e^{t\hat L}\expc \left[e^{t\textsf L}\big|\right],\quad \textsf L=\sum_{2\le 2k\le N} \eta_{k,|}
\ee
Here $\expc \left[\cdot|\right]$ stands for the conditional expectation against the variables 
$$\left(U_2, ..., U_N,\textsf x_{1,|},...,\textsf x_{N-1,|}\right)$$
and maybe also on $r_{1,|}$ if $N$ is even.

Equation \eqref{eq:Varlogr} in the statement of Corollary~\ref{cor:logrindep} says
\be
\nn\Var \left[\eta_{k,|}\right]\ge {c\over W} \qquad\qquad\\
\label{eq:Bernsteinlogr}\\
\nn\left|\eta_{k,|}\right|\le C\log W,\qquad\qquad
\ee
That implies, in particular, condition \eqref{eq:Bernsteicond} in Lemma~\ref{lem:greenorlowmoments} with the same $H=C\log W$ for all $k$. Corollary~\ref{cor:classicalexpest} applies:
\be\label{eq:Bernsteinesft}
\Big[\expc \big[e^{t\textsf L}\big|\big]\Big]^2\le\expc \big[e^{2t\textsf L}\big|\big]\times\prod_{2 \leq 2k \leq N-1}\Big(1-{t^2\sigma_{|,k}^2\over 4}\Big),
\ee
for any $|t|\le c H^{-1}$, where $\sigma_{|,k}^2=\Var \left[\eta_{k,|}\right]$ stands for the conditional variance, $0<c\ll 1$ is an absolute constant. Combine with \eqref{eq:Bernsteinlogr}, get
\be\label{eq:BernsteinesfL}
\Big[\expc \big[e^{t\textsf L}\big|\big]\Big]^2\le\expc \big[e^{2t\textsf L}\big|\big]\times e^{-{c_1N\over W\left(\log W\right)^2}},
\ee
where $c_1>0$ is an absolute constant. Combine \eqref{eq:Bernexpest} with \eqref{eq:BernsteinesfL} and with Cauchy-Schwarz, get 
\begin{gather*}
\left[\expc \left[\left\|G_{[1,N]}\left(1;N\right)g\right\|^t\right]\right]^2=\left[\expc \left[e^{t\hat L}\expc \big[e^{t\textsf L}\big|\big]\right]\right]^2\le\\
\left[\expc \left[e^{t\hat L}\left[\expc \big[e^{2t\textsf L}\big|\big]\right]^{{1\over 2}}\times e^{-{c_1N\over 2W\left(\log W\right)^2}}\right]\right]^2=
e^{-{c_1N\over W\left(\log W\right)^2}}\left[\expc \left[e^{t\hat L}\left[\expc \big[e^{2t\textsf L}\big|\big]\right]^{{1\over 2}}\right]\right]^2\le\\
e^{-{c_1N\over W\left(\log W\right)^2}}\expc \left[e^{2t\hat L}\expc \big[e^{2t\textsf L}\big|\big]\right]=
e^{-{c_1N\over W\left(\log W\right)^2}}\expc \left[e^{2t\hat L}e^{2t\textsf L}\right]=\\
e^{-{c_1N\over W\left(\log W\right)^2}}\expc \left[\left\|G_{[1,N]}\left(1;N\right)g\right\|^{2t}\right]
\end{gather*}
Combine with \eqref{eq:Wegner}, get
\be\label{eq:Wegnerg}
\expc \left[\left\|G_{[1,N]}\left(1;N\right)g\right\|^t\right]\le C_1W^{C_2}\times e^{-{c_1N\over 2W\left(\log W\right)^2}}\le e^{-{cN\over W\left(\log W\right)^2}}
\ee
Take here $t={C\over \log W}$, combine with Chebishev inequality, get 
\begin{gather*}
\pr\left\{\left\|G_{[1,N]}\left(1;N\right)g\right\|>M\right\}=\pr\left\{\left\|G_{[1,N]}\left(1;N\right)g\right\|^t>M^t\right\}\le \\
{1\over M^t}\times\expc \left[\left\|G_{[1,N]}\left(1;N\right)g\right\|^t\right]\le {1\over M^t}\times e^{-{cN\over W\left(\log W\right)^2}},\quad \text {for any $M>0$},\\
\pr\left\{\left\|G_{[1,N]}\left(1;N\right)g\right\|>e^{-{cN\over 2W\left(\log W\right)^3}}\right\}\le e^{-{cN\over 2W\left(\log W\right)^2}}
\end{gather*}
Recall that here $\pr$ stands not for the original probability $\pr_N$, but for $\pr_N$ conditioned on the set $\cG_N$ which is the complement of $\widehat{\cB}$, where 
\be\label{eq:inverscrudef}
\pr_N\left[\widehat{\cB}\right]\le CNW^{D-1},
\ee
Recall that in Theorem C we assume that $W^2<N\le W^{D_0}$ where $D_0\gg 1$ is a constant. Take $D_0=c_0D$ with $0<c\ll 1$ being an absolute constant. Combine, run $g=e_j$ where $e_j$ is the standard basis in $R^W$, add up, get the statement in Theorem C.
\end{proof}

\appendix
\section{Matrix Variables Orthogonal Change.}\label{appedndixA}
\begin{defi}\label{defi:Qid}
$(1)$ Denote $\textsf M_{\textsf s}\left(W\right)$ the linear space of all real symmetric matrices. Denote $\textsf I(V)=\textsf V$ the identification 
\begin{gather*}
\textsf M_{\textsf s}\left(W\right)\ni V=(v_{p, q})_{1 \leq p,q \leq W}\to \textsf V:=(v_{p, q})_{1 \leq p\le q \leq W}:=\\
(v_{p, p})_{1 \leq p\leq W}\oplus(v_{p, q})_{1 \leq p<q \leq W}\in \IR^W\bigoplus\IR^{{W(W-1)\over 2}}=\IR^D,\quad D={W(W-1)\over 2}
\end{gather*}

\smallskip
$(2)$ Take an orthogonal matrix $Q$. Consider the maps
\begin{gather*}
\Phi_Q:\textsf M_{\textsf s}\left(W\right)\ni V\to QVQ^t\in \textsf M_{\textsf s},\\
\tilde \Phi_Q:\IR^D\ni \textsf V\to \textsf I\circ\Phi_Q\circ\textsf I^{-1}(\textsf V)\in \IR^D
\end{gather*}
\end{defi} 
\begin{lemma}\label{lem:elementary3} Using the notations in Definition~\ref{defi:Qid} the Jacobian $J_{\tilde\Phi_Q}$ of the map $\tilde\Phi_Q$ is $=1$.
\end{lemma}
\begin{proof} Consider the Frobenius inner product in $\textsf M_{\textsf s}\left(W\right)$ 
\begin{gather*}
\langle V,U\rangle_{F}=Tr\left(U^tV\right)
\end{gather*}
Use the identification $\textsf I(V)=\textsf V$, define the inner product
\begin{gather*}
\langle \textsf V,\textsf U\rangle_{F}:=\langle \textsf I^{-1}\textsf V,\textsf I^{-1}\textsf U\rangle_{HS}
\end{gather*}
The map $\Phi_Q(V)=Q^tVQ$, $V\in\textsf M_{\textsf s}\left(W\right)$  is linear and preserves the Frobenius inner product in $\textsf M_{\textsf s}\left(W\right)$.  Conclude $\tilde\Phi_Q(V)$ is a linear isometric map in the Euclidean space $\IR^D,\langle ,\rangle_{F}$. Enumerate the standard basis in $\IR^W$, resp.  $\IR^{{W(W-1)\over 2}}$, as $e_{p,p}$, $1\le p\le W$, resp. $e_{p,q}$, $1\le p<q\le W$. Use this enumeration for the standard basis in $\IR^D=\IR^W\bigoplus\IR^{{W(W-1)\over 2}}$, denote $v_{p,p}$, $1\le p\le W$, resp. $v_{p,q}$, $1\le p<q\le W$ the respective coordinates. Use the definitionof the trace, conclude that
\begin{gather*}
\langle \textsf V,\textsf V\rangle_{F}=\sum_{1\le p\le W}v^2_{p,p}+2\sum_{1\le p<q\le W}v^2_{p,q}
\end{gather*}
The map $\tilde\Phi_Q(V)$ preserves the norm $\left\|\textsf V\right\|_{F}:=\sqrt{\langle \textsf V,\textsf V\rangle_{F}}$. 
In particular, the map preserves the volume $\mathcal {V}_{HS}$ in the Euclidean space $\IR^D,\langle ,\rangle_{HS}$. Compare the inner product $\langle ,\rangle_{HS}$ against the standard inner product $\langle ,\rangle$ in $\IR^D$ conclude $\mathcal {V}_{HS}=C(W)\mathcal {V}$, where $\mathcal {V}$ stands for the standard volume in $\IR^D$, $C(W)=2^{{W(W-1)\over 4}}$. 
Conclude the map $\tilde\Phi_Q(V)$ preserves the standard volume in $\IR^D$, $\tilde J_{\Phi_Q}=1$.

\end{proof}

\end{document}